\documentclass{theoretics}
\title{Testing Support Size More Efficiently Than Learning Histograms}
\ThCSshorttitle{Testing Support Size More Efficiently Than Learning Histograms}

\ThCSauthor[1]{Renato Ferreira Pinto Jr.}{rf2985@columbia.edu}[0009-0003-2346-8423]
\ThCSauthor[2]{Nathaniel Harms}{nharms@cs.ubc.ca}[0000-0003-0259-9355]
\ThCSaffil[1]{Columbia University, USA}
\ThCSaffil[2]{University of British Columbia, Canada}

\ThCSshortnames{R. Ferreira Pinto Jr. and N. Harms}

\ThCSkeywords{Distribution testing, property testing, support size, Chebyshev polynomials}
 \ThCSthanks{%
     This work was done while Renato Ferreira Pinto Jr.\ was a student at the
     University of Waterloo, and Nathaniel Harms was a postdoc at EPFL.
     Renato Ferreira Pinto Jr.\ was supported by an NSERC Canada Graduate
     Scholarship. Nathaniel Harms was supported by the Swiss State Secretariat
     for Education, Research, and Innovation (SERI) under contract number
     MB22.00026, an NSERC Postdoctoral Fellowship, and a Simons Fellowship. 
     Thanks to Amit Levi for a discussion on the bounded support size
     assumption, and to the anonymous reviewers for their comments and
     suggestions. Some of this work was done during the Sublinear Algorithms
     program at the Simons Institute. We also thank Michael Kapralov for funding
     a research visit for Renato to EPFL.
     A preliminary version of this work appeared at STOC 2025.
}

\ThCSyear{2026}
\ThCSarticlenum{10}
\ThCSreceived{Oct 16, 2025}
\ThCSrevised{Mar 1, 2026} 
\ThCSaccepted{Mar 21, 2026}
\ThCSpublished{May 21, 2026}
\ThCSdoicreatedtrue

\makeatletter
\newcommand\IfRestateTF{%
  \ifx\label\thmt@gobble@label 
    \expandafter\@firstoftwo
  \else
    \expandafter\@secondoftwo
  \fi
}
\makeatother
\newcommand{\RestateRemark}{\IfRestateTF{{\normalfont\bfseries (Restated) }}{}}

\usepackage{thmtools}
\usepackage[capitalise,noabbrev,nameinlink]{cleveref}
\usepackage{mathtools}
\usepackage{xspace}
\usepackage{verbatim}
\usepackage{mathrsfs}
\usepackage{pgf}
\usepackage{tabularx}
\usepackage{derivative}
\usepackage{bm}
\usepackage{multirow}
\usepackage{diagbox}
\usepackage{nicematrix}

\usepackage{dsfont}

\newcommand*\ie{i.\kern.1em e.\ }
\newcommand*\eg{e.\kern.1em g.\ }
\newcommand*\cf{c.\kern.1em f.\ }
\newcommand*\almev{a.\kern.1em e.\ }

\definecolor{ama-iro}{RGB}{0, 158, 243.0}
\definecolor{fuyu-gaki}{RGB}{251, 74, 52}
\definecolor{momiji}{RGB}{245, 70, 111}
\definecolor{hotaru-bi}{RGB}{229,221,58} 
\definecolor{kon-peki}{RGB}{1,120,217}
\definecolor{shin-kai}{RGB}{77,98,152}
\definecolor{shin-ryoku}{RGB}{1,145,97}
\definecolor{yama-budo}{RGB}{171,14,122}

\ThCSnewtheoita{assumption}
\ThCSnewtheoita{question}
\ThCSnewtheoita{hypothesis}
\crefname{observation}{Observation}{Observations}
\crefname{claim}{Claim}{Claims}
\crefname{fact}{Fact}{Facts}

 \ThCSnewtheostd{condition}
  \ThCSnewtheostd{goal}

  \declaretheorem[style=standard,name=Constraint,numbered=yes]{constraint}

\renewcommand{\theconstraint}{%
  \ifnum\value{constraint}=5
    IVb%
  \else
    \Roman{constraint}%
  \fi
}

\crefname{constraint}{Constraint}{Constraints}

\newcommand{\ignore}[1]{}

\DeclareMathOperator{\supp}{supp}   
\DeclareMathOperator{\poly}{poly}

\newcommand{\dist}{\mathsf{dist}}

\newcommand{\Var}[1]{\mathrm{Var} \left[ #1 \right]}

\newcommand{\Ex}[1]{\bE \left[ #1 \right]}

\renewcommand{\Pr}[1]{\bP \left[ #1 \right]} 
\newcommand{\Pru}[2]{\underset{ #1 }\bP \left[ #2 \right]}
\newcommand{\Pruc}[3]{\underset{ #1 }\bP \left[ #2 \;\; \mid \;\; #3 \right]}

\newcommand*{\define}{\mathrel{\vcenter{\baselineskip0.5ex \lineskiplimit0pt
                      \hbox{\scriptsize.}\hbox{\scriptsize.}}}%
                      =}

\newcommand{\ceil}[1]{\ensuremath{\lceil #1 \rceil}}

\DeclarePairedDelimiter{\abs}{\lvert}{\rvert}

\newcommand{\ind}[1]{\mathds{1} \left[ #1 \right] }

\newcommand{\zo}{\{0,1\}}

\newcommand{\cH}{\ensuremath{\mathcal{H}}}

\newcommand{\cX}{\ensuremath{\mathcal{X}}}

\newcommand{\bE}{\ensuremath{\mathbb{E}}}

\newcommand{\bN}{\ensuremath{\mathbb{N}}}
\newcommand{\bP}{\ensuremath{\mathbb{P}}}
\newcommand{\bR}{\ensuremath{\mathbb{R}}}

\usepackage{subcaption}

\usepackage{etoolbox}

\newcommand{\Poi}{\mathsf{Poi}}

\newcommand{\ACCEPT}{\mathsf{Accept}}
\newcommand{\REJECT}{\mathsf{Reject}}
\newcommand{\successprob}{\sigma}

\newcommand{\VC}{\mathsf{VC}}
\newcommand{\TV}{\mathsf{TV}}

\newcommand{\distTV}{\dist_\TV}
\newcommand{\DIST}{\mathsf{DIST}}
\newcommand{\FUN}{\mathsf{FUN}}
\newcommand{\eff}{\mathsf{eff}}

\renewcommand{\epsilon}{\varepsilon}

\begin{document}

\date{}

\maketitle

\begin{abstract}
Consider two problems about an unknown probability distribution $p$:
\begin{enumerate}
\item How many samples from $p$ are required to test if $p$ is
supported on $n$ elements or not? Specifically, given samples from $p$, determine
whether it is supported on at most $n$ elements, or it is ``$\epsilon$-far'' (in total
variation distance) from being supported on $n$ elements.
\item Given $m$ samples from $p$, what is the largest lower bound on its support size
that we can produce?
\end{enumerate}
The best known upper bound for problem (1) uses a general algorithm for learning the histogram of
the distribution $p$, which requires $\Theta(\tfrac{n}{\epsilon^2 \log n})$ samples. We show that
testing can be done more efficiently than learning the histogram, using only $O(\tfrac{n}{\epsilon
\log n} \log(1/\epsilon))$ samples, nearly matching the best known lower bound of
$\Omega(\tfrac{n}{\epsilon \log n})$. This algorithm also provides a better solution to problem (2),
producing larger lower bounds on support size than what follows from previous work. The proof relies
on an analysis of Chebyshev polynomial approximations \emph{outside} the range where they are
designed to be good approximations, and the paper is intended as an accessible self-contained
exposition of the Chebyshev polynomial method.
\end{abstract}



\section{Introduction}

Sadly, it is often necessary to make decisions based on probability distributions. For example, you
may need to make a decision based on the support size of an unknown distribution:

\begin{example}
\label{example:forest}
There is a population of fish in the lake, which may be described as a probability distribution
$p$ over the set of possible fish species. You want to preserve a specimen of each species.
You have no idea how many fish or fish species there are, and you only have
10000 sample jars. Is that enough, or should you buy more? You would like to decide between: 
\begin{enumerate}[itemsep=0pt,topsep=0.3em]
\item There are at most 10000 species of fish in the lake. (You have enough jars.)
\item Any collection of 10000 species will miss at least 0.1\% of the population, \ie 0.1\% of the
fish will belong to species \emph{not} in the collection. (You need more jars.)
\end{enumerate}
\ignore{You are happy with specimens representing 99.9\% of the \emph{fish}, as opposed to 99.9\% of
the \emph{species}, since the vast majority of species may comprise only a tiny fraction of fish.}
The na\"ive strategy is to sample random fish until you have filled all the jars, and if you
continue to see new species, buy more jars. Of course, you prefer to predict in advance if you need
more jars. How many fish do you need to sample before you can make this decision (and be correct
with probability at least 99\%)?  
\end{example}
The best known algorithms for making this type of decision work by learning the histogram of the
distribution (the unordered multiset of nonzero probabilities densities $\{
p_i | p_i \neq 0\}$). In this paper we show there is a more sample-efficient way to do it. Formally,
the problem is to distinguish between distributions with support size at most $n$, and those which
are \emph{$\epsilon$-far} from having support size at most $n$. We say $p$ is \emph{$\epsilon$-far}
from having support size at most $n$ when, for every probability distribution $q$ with support size
$|\supp(q)| \leq n$, the total variation (TV) distance between $p$ and $q$ is $\distTV(p,q) >
\epsilon$. Then:
\begin{definition}[Testing Support Size]
\label{def:testing-support-size}
A \emph{support-size tester} with sample complexity $m(n,\epsilon,\successprob)$ is an algorithm $A$
which takes as input the parameters $n \in \bN$, $\epsilon \in (0,1)$, and $\successprob \in (0,1)$.
It draws a multiset $\bm S$ of $m = m(n,\epsilon,\successprob)$ independent samples from an unknown
probability distribution $p$ over $\bN$, and its output must satisfy:
\begin{enumerate}[itemsep=0pt, topsep=0.3em]
\item If $|\supp(p)| \leq n$ then $\Pru{\bm S}{A(\bm S) \text{ outputs } \ACCEPT} \geq \successprob$; and
\item If $p$ is \emph{$\epsilon$-far} from having support size at most $n$, then
$\Pru{\bm S}{A(\bm S) \text{ outputs } \REJECT} \geq \successprob$.
\end{enumerate}
Unless otherwise noted, we set $\successprob = 3/4$.
\end{definition}

Testing support size is a basic statistical decision problem (see \eg the textbook \cite[\S 11.4]{Gol17} and
recent work \cite{GR23,AF24,AFL24,KLR25}) that underlies many other commonly studied tasks.  For example,
it is a decision version of the problem of \emph{estimating} support size -- a problem whose history
dates back to Fisher \cite{FCW43}, Goodman \cite{Goodman49}, and Good \& Turing \cite{Good53}; see
\cite{BF93} for a survey, and more recent work including \cite{RRSS09,VV11stoc,VV17,WY19,HO19}.
Testing support size is also
important for understanding the \emph{testing vs.~learning} question in property testing.

\subsection{Testing \texorpdfstring{vs.~Learning Distributions}{vs. Learning Distributions}}

Testing support size is a \emph{distribution testing} problem, which
is a type of \emph{property testing}.  A basic technique in property testing is
the testing-by-learning approach of \cite{GGR98}, where the tester learns an approximation of
the input and makes its decision based on this approximation. One of the guiding questions in
property testing is: When can this technique be beaten?

It requires $\Theta(n/\epsilon^2)$ samples to \emph{learn} a distribution $p$ over domain $[n]$, up
to TV distance $\epsilon$ (see \eg the survey \cite{Can20}). Remarkable and surprising recent work
\cite{VV11stoc,VV17,HJW18,HO19pml} proves that the \emph{histogram} of distributions $p$ over $[n]$
can be learned with only $\Theta(\tfrac{n}{\epsilon^2 \log n})$ samples --- a vanishing fraction of
the domain. This means one can test the support size -- or perform several other tasks -- using
$O(\tfrac{n}{\epsilon^2 \log n})$ samples to learn the histogram.  As noted in the textbook
\cite{Gol17} and recent work \cite{GR23,AFL24}, this testing-by-learning algorithm gives the best
known upper bound for testing support size when the true support size is promised to satisfy
$|\supp(p)| = O(n)$. Without this promise there does not appear to be a simple way to obtain a
similar bound from known results (see \cite[pp.21]{GR23} and the discussion in
\cref{section:support-assumption}, where we sketch arguments to obtain bounds of
$O(\tfrac{n}{\epsilon^3 \log n})$ from \cite{VV16}). We show that the testing-by-learning algorithm
can be beaten, while also removing any restriction on the true support size:
\begin{theorem}
\label{res:intro-main}
For all $n \in \bN$ and $\epsilon \in (0,1)$, the sample complexity of testing support size of an
unknown distribution $p$ (over any countable domain) is at most
\[
  m(n,\epsilon) = O\left(\frac{n}{\epsilon \log n} \cdot \min\left\{ \log(1/\epsilon), \log n \right\} \right) \,.
\]
\end{theorem}
In terms of \cref{example:forest}, if you have $n$ jars, you can decide whether to buy more jars
after taking $O(\tfrac{n}{\epsilon \log n}\log(1/\epsilon))$ samples, filling
less than a sublinear $O(\tfrac{\log(1/\epsilon)}{\epsilon \log n})$-fraction of the jars.
\cref{res:intro-main} nearly matches the best known lower bound of $\Omega(\tfrac{n}{\epsilon \log
n})$ (which can be deduced from lower bounds of \cite{VV11stoc,WY19}). An important note about this
problem is that, in the setting of \cref{res:intro-main}, the related problem of \emph{estimating}
the support size is impossible: the unknown support size of $p$ is unbounded and the probability
densities $p_i$ can be arbitrarily small. If one makes assumptions to avoid these issues, then there
are tight bounds of $\Theta(\tfrac{k}{\log k} \log^2(1/\epsilon))$ for estimating the support size
up to $\pm \epsilon k$ when each nonzero density $p_i$ is promised to satisfy $p_i > 1/k$
\cite{WY19,HO19} (see also \cite{RRSS09,VV11stoc,ADOS17,VV17}). Recalling \cref{example:forest}, $k$
would be the total number of \emph{fish} in the lake, whereas in our results $n$ is the number of
\emph{species}.

\subsection{Good Lower Bounds on Support Size}
Quoting I.~J.~Good in \cite{BF93}, ``I don't believe it is usually possible to estimate the number
of unseen species [\ie support size]\dots but only an approximate lower bound to that number.'' So,
in lieu of an estimate, what is the best lower bound we can get?
\begin{question}
\label{question:good-lb}
Given $m$ samples from a distribution $p$ and parameter $\epsilon \in (0,1)$, what is the biggest
number $\widehat S$ that we can output, while still satisfying $\widehat S \leq
(1+\epsilon)|\supp(p)|$?
\end{question}
Our algorithm can produce lower bounds as large as $\Omega(\tfrac{\epsilon}{\log(1/\epsilon)} m \log
m)$ out of only $m$ samples. Let us formalize the quality of these lower bounds.  A reasonable
target for such a ``Good'' lower bound $\widehat S$ is that it should exceed the
\emph{$\epsilon$-effective support size}, the smallest number of elements covering $1-\epsilon$
probability mass:

\begin{definition}[Effective support size]
\label{def:effective}
For any probability distribution $p$ over $\bN$ and any $\epsilon \in (0,1)$, we define the
\emph{$\epsilon$-effective support size} $\eff_\epsilon(p)$ as the smallest number $k \in \bN$
such that there exists distribution $q$ with $|\supp(q)|=k$ and $\dist_\TV(p,q) \leq \epsilon$.
\end{definition}

\noindent
With no assumptions on the distribution, the effective support size is a more natural target for
estimation, because it accounts for the fact that an arbitrarily large number of elements in the
support may comprise an arbitrarily small probability mass. Effective support size was used in
\cite{CDS18,BCG19}, and algorithms for estimating it were analyzed in \cite{Gol19vdf,NT23,Gol25},
but these algorithms assume the algorithm can learn the probability densities $p_i$ by quering
element $i$. Our algorithm provides a Good lower bound using only samples:

\begin{corollary}
\label{res:intro-good-lb}
For all $n \in \bN$ and $\epsilon \in (0,1)$, there is an algorithm which draws at most
\[
  O\left(\frac{n}{\epsilon \log n} \cdot \min\{ \log(1/\epsilon), \log n \} \right)
\]
samples from an arbitrary distribution $p$, and outputs a number $\widehat S$ which satisfies (with
probability at least $3/4$)
\[
  \min\{ \eff_\epsilon(p), n \} \leq \widehat S \leq (1+\epsilon) |\supp(p)| \,.
\]
\end{corollary}
\noindent
We are not aware of any prior work focusing on this type of guarantee, although it seems
quite natural: the value $\widehat S$ output by the algorithm can be used for problems like
\cref{example:forest} to predict the resources required for a task depending on support size. The
algorithm will tell us either:
\begin{enumerate}[itemsep=0pt]
\item If $\widehat S < n$ (which would be the case if we were promised, say, $|\supp(p)| < n/2$),
then $\widehat S$ units (jars) suffice to cover the distribution up to the removal of $\epsilon$
mass, and fewer than $\approx \widehat S$ units will fail to cover the whole distribution.
Furthermore, we obtain this estimate using fewer samples than required to learn the histogram.
\item If $\widehat S \geq n$, then we have learned that we need more than $\approx \widehat S$ units
to cover the distribution.
\end{enumerate}
Let us translate this result into an answer for \cref{question:good-lb}. For a fixed sample size
$m$, the na\"ive approach, taking $\widehat S$ to be the number of unique elements appearing in the
sample, produces a lower bound $\widehat S \leq |\supp(p)|$ that exceeds $\min\{\eff_\epsilon(p),
\Theta(\epsilon m)\}$ (see \cref{res:naive-estimator}). The histogram learners (if they were proved
to succeed on this task without an assumption on true support size) produce a lower bound $\widehat
S \leq (1+\epsilon)|\supp(p)|$ that exceeds $\min\{\eff_\epsilon(p), \Theta(\epsilon^2 m \log m)\}$.
Our algorithm improves this lower bound up to $\min\{\eff_\epsilon(p),
\Theta(\tfrac{\epsilon}{\log(1/\epsilon)} m \log m)\}$.  

\subsection{Testing \texorpdfstring{vs.~Learning Boolean Functions}{vs. Learning Boolean Functions}}

Testing support size of distributions is also important for understanding testing vs.~learning of
Boolean functions. Arguably the most well understood model of learning is the distribution-free
sample-based PAC model: The algorithm receives samples of the form $(x, f(x))$ where $x \sim p$ is
drawn from an unknown distribution, and labeled by an unknown function $f : \cX \to \zo$, which is
promised to belong to a hypothesis class $\cH$. The algorithm should output
a function $g \in \cH$ such that $\mathbb{P}_{x \sim p}[ f(x) \neq g(x) ] \leq \epsilon$. (Requiring
$g \in \cH$ is called \emph{proper} PAC learning.)

Whereas the PAC learner assumes the input $f$ belongs to the hypothesis class $\cH$, a \emph{tester}
for $\cH$ is an algorithm which tests this assumption: 

\begin{definition}[Distribution-free sample-based testing; see formal \cref{def:testing-functions}]
Given labeled samples $(x, f(x))$ with $x \sim p$ drawn from unknown distribution $p$ and labeled by
unknown function $f$, decide whether (1) $f \in \cH$ or (2) $f$ is $\epsilon$-far from all functions
$g \in \cH$, meaning $\Pru{x \sim p}{f(x) \neq g(x)} \geq \epsilon$.
\end{definition}

Testing $\cH$ can be done by running a proper PAC learner and checking if it worked \cite{GGR98}.
One motivation for testing algorithms is to aid in model selection, \ie choosing an appropriate
hypothesis class $\cH$ for learning. To be useful for this, the tester should be more
efficient than the learner. The sample complexity of proper PAC learning is between
$\Omega(\VC/\epsilon)$ and $O( \tfrac{\VC}{\epsilon} \log(1/\epsilon) )$, where $\VC$ is the VC
dimension of $\cH$ (and there are examples where either bound is tight) \cite{BEHW89,Han16,Han19}.
But there is no similar characterization of the sample complexity for \emph{testing}, and it is not
known when it is possible to do better than the testing-by-PAC-learning method.

In fact, very little is known about property testing in
the distribution-free sample-based setting, and there are not many positive results (see \eg
\cite{GR16,Gol17,RR20,BFH21,RR22}). For example, the basic class $\cH_n$ of functions $f \colon \bN
\to \zo$ with $|f^{-1}(1)| \leq n$ has a tight bound of $\Theta(\VC/\epsilon)$ for proper learning
with samples, but for testing with samples, tight bounds in terms of both $\VC$ and $\epsilon$ are
not known\footnote{If the domain $\bN$ is replaced with $[Cn]$ for constant $C$, and queries are
allowed, then $O(1/\epsilon^2)$ queries suffice.}\!\!. Using the lower bounds of \cite{VV17,WY19}
for support size estimation, \cite{BFH21} show lower bounds of the form $\Omega(\tfrac{\VC}{\epsilon
\log \VC})$ for this class and several others (\eg halfspaces, $k$-alternating functions \&c.). This
still leaves room for interesting upper bounds: any bound of $o(\VC/\epsilon)$ means that the tester
is doing something significantly different from testing-by-PAC-learning.

Indeed, \cite{GR16} show an upper bound of $O(\tfrac{\VC}{\epsilon^2 \log \VC})$ samples for testing
$\cH_n$ (with the extra promise that $f^{-1}(1) \subset [2n]$), by learning the \emph{distribution}
instead of the \emph{function}: their algorithm learns the histogram of the underlying
distribution $p$ on the subdomain $f^{-1}(1)$, with some adjustments to handle the fact that the
probability mass of $p$ inside $f^{-1}(1)$ may be small. This shows that the lower bound of
\cite{BFH21} from support-size estimation is sometimes tight in terms of the VC dimension, and hints
at a closer connection to distribution testing.  We give a simple proof of a tighter and cleaner
relationship -- that testing support size of distributions and functions are essentially the same
problem.

\begin{theorem}[Equivalence of testing support size for distributions and functions]
\label{res:intro-functions}
Let $m^\DIST(n,\epsilon,\successprob)$ be the sample
complexity for testing support size with success probability $\successprob$.
Let $m^\FUN(n,\epsilon, \successprob)$ be
the sample complexity for distribution-free sample-based testing $\cH_n$ with success
probability $\successprob$.
Then $\forall$ $n \in \bN$, $\epsilon, \successprob \in (0,1)$, and $\xi \in (0, 1-\successprob)$,
\[
  m^\DIST(n,\epsilon,\successprob) \leq m^\FUN(n,\epsilon,\successprob)
  \leq m^\DIST(n,\epsilon,\successprob+\xi) + O\left(\tfrac{\log(1/\xi)}{\epsilon}\right) \,.
\]
\end{theorem}

With \cref{res:intro-main}, this improves the best upper bound for testing $\cH_n$ from
$O(\tfrac{\VC}{\epsilon^2 \log \VC})$ \cite{GR16} to $O(\tfrac{\VC}{\epsilon \log \VC}
\log(1/\epsilon))$. So, not only can $\cH_n$ be tested more efficiently than learned, but
even more efficiently than the histogram of the
underlying distribution can be learned.

We note also that there are reductions from testing the class $\cH_n$ to testing several other
classes like $k$-alternating functions, halfspaces, and decision trees \cite{BFH21,FH23}, which
emphasizes the role of testing support size of distributions to the testing vs.~learning question
for Boolean functions.

\subsection{Techniques and Open Problems}

The first part of our proof closely follows the optimal upper bound of \cite{WY19} for the support
size \emph{estimation} problem. In short, if domain element $i$ has probability density $p_i$ and
appears $N_i$ times in the sample, the estimator of \cite{WY19} outputs $\sum_i f(N_i)$, where
$f(N_i)$ is constructed such that $\Ex{f(\bm N_i)} = Q(p_i)$ where $Q(p_i)$ is a function satisfying
$Q(p_i) \approx \ind{p_i > 0}$. If this can be done, then
\[
    \Ex{ \sum_i f(\bm N_i) } = \sum_i Q(p_i) \approx |\supp(p)| .
\]
They make the necessary assumption that every nonzero density satisfies $p_i > 1/n$, allowing them
to use Chebyshev polynomials to define a suitable $Q$ where $Q(p_i) \approx \ind{p_i > 1/n}$. 
\cref{section:test-statistic} explains why Chebyshev polynomials arise naturally.

We follow the same strategy, although we no longer have any assumption on the nonzero densities.  We
construct a Chebyshev polynomial approximation $Q(p_i) \approx \ind{p_i > \ell}$ for carefully
chosen $\ell$ (depending on $\epsilon$ and $n$). The main challenge is to account for the densities
$0 < p_i < \ell$, where the Chebyshev polynomial is not ``designed'' to be a good approximation, \ie
$Q(p_i) \not\approx \ind{p_i > 0}$.

\paragraph*{Main idea:} We establish a tradeoff between the number of ``heavy elements'' with $p_i >
\ell$ (on which the polynomial approximation satisfies $Q(p_i) \approx 1$), and the value of the
estimator on the ``light elements'' with $p_i \leq \ell$. This is proved in
\cref{section:improved-bound}, with an easier version giving a worse bound in
\cref{section:correctness_and_parameters}. The easier version uses concavity of $Q$ on $(0, \ell)$
to show $Q(p_i) \gtrsim p_i/\ell$ (\cref{res:q-light-bounds}).  In this case, we use $\ell \approx
\epsilon/n$.  This gives us a tradeoff of the form
\begin{equation}
\label{eq:techniques-linear}
    \sum_i Q(p_i) \approx  \#\{\text{heavy elements}\} + \frac{n}{\epsilon} \sum_{\text{light }
i} p_i .
\end{equation}
Then, what remains to show is that, when $p$ is $\epsilon$-far from $|\supp(p)| \leq n$, this
quantity is large (and then bound the variance of the estimator, which depends on the degree of the
polynomial approximation). To get a better bound on the sample size and prove our main theorem, we
use a different setting of parameters in our polynomial approximation, and we establish a more
complicated tradeoff than \cref{eq:techniques-linear} (see \cref{res:Q*-worst-case}), essentially
\begin{equation}
    \sum_i Q(p_i) \geq \left(n + \frac{\epsilon}{p^*}\right) Q(p^*),
\end{equation}
where $p^*$ is the $n^{th}$ largest density and $\epsilon$ is the distance to $|\supp(p)|\leq n$,
and we characterize the worst-case behaviour of this tradeoff.

\paragraph*{Open questions.}

The sample complexity upper bound from \cref{res:intro-main} leaves a small gap compared to the best
known lower bound of $\Omega(\tfrac{n}{\epsilon \log n})$ from \cite{VV11stoc,WY19}. Can we close
this gap? We remark that our proof strategy does not allow us to improve any of our parameters (see
\cref{remark:no-better-parameter}), and therefore we suspect that our upper bound is tight, unless
there is a better algorithmic strategy (say, a non-linear estimator).

We may also ask for a \emph{tolerant} tester \cite{PRR06}, \ie an algorithm that accepts not only
when $p$ has support size at most $n$, but also when $p$ is \emph{$\epsilon'$-close} to such a
distribution for some $\epsilon' < \epsilon$. Equivalently, we wish to distinguish between
$\eff_{\epsilon'}(p) \le n$ and $\eff_\epsilon(p) > n$. Our tester from \cref{res:intro-main}, like
any sample-based tester, enjoys the very small amount of tolerance $\epsilon' = \Theta(1 /
m(n,\epsilon))$. Under the promise that the true support size satisfies $|\supp(p)| = O(n)$,
algorithms for learning the histogram \cite{VV11stoc,VV17,HJW18,HO19pml} imply an upper bound of
$O\left(\tfrac{n}{(\epsilon-\epsilon')^2 \log n}\right)$ for tolerant testing of support size, and
in \cref{section:support-assumption} we discuss how the techniques of \cite{VV16} may plausibly
allow one to drop the support size assumption at the cost of higher dependence on
$\epsilon-\epsilon'$. Can we achieve $\tfrac{n}{\log n} \cdot \widetilde
O\left(\tfrac{1}{(\epsilon-\epsilon')^2}\right)$?

\subsection{Organization}

One may treat the first part of our paper (up to \cref{remark:recover-wy}) as an
alternative self-contained exposition of \cite{WY19}, and we intend the paper to be as clear and
accessible an explanation of the Chebyshev polynomial method as possible (we also refer the reader
to the textbook \cite{WY20} on polynomial methods in statistics).

\begin{description}[topsep=0pt,itemsep=0em]
\item[\cref{section:test-statistic}] sets up the testing algorithm using generic parameters.
\item[\cref{section:parameters}] places
\cref{constraint:d_log,constraint:right_tail,constraint:variance,constraint:ell} on the parameters
in order to guarantee the success of the tester, and then in
\cref{section:correctness_and_parameters} we optimize the parameters to complete a weaker (but
simpler) version of \cref{res:intro-main}. The stronger version follows by replacing
\cref{constraint:ell} with the looser \cref{constraint:better_ell} in \cref{section:improved-bound}.
\item[\cref{section:good-lower-bound}] proves that the algorithm outputs a lower bound on
support size, \cref{res:intro-good-lb}.
\item[\cref{section:functions}] proves equivalence of testing support size of
distributions and functions, \cref{res:intro-functions}.
\end{description}

\section{Defining a Test Statistic}
\label{section:test-statistic}

We begin by defining a ``test statistic'' $\widehat S$ which the tester will use to make its
decision -- the tester will output \textsf{Accept} if $\widehat S$ is small and output
\textsf{Reject} if it is large. But our test statistic will require $\log(1/\epsilon) < \log(n)$, so
we first handle the case where the parameter $\epsilon$ is very small.

\subsection{Small \texorpdfstring{$\epsilon$}{epsilon}}
\label{section:small-epsilon}

If $\epsilon$ is small enough that $\log(1/\epsilon) = \Omega(\log n)$, we will use the following
simple tester. This result is folklore (appearing \eg in \cite{GR23,AF24,AFL24}) but we include a
proof for the sake of completeness. We start with the following simple observation:
\begin{observation}
    \label{obs:eff}
    For any $n \in \bN$, $\epsilon \in (0,1)$, and probability distribution $p$, $p$ is
    $\epsilon$-far from having support size at most $n$ if and only if $\eff_\epsilon(p) > n$.
\end{observation}

\noindent
Using this observation, we give

\begin{proposition}
\label{res:naive-estimator}
There is an algorithm which, given inputs $n \in \bN$, $\epsilon \in (0,1)$, and sample access to an
arbitrary distribution $p$, draws at most $O(n/\epsilon)$ samples from $p$ and outputs a number
$\widehat{\bm S}$ which (with probability at least $3/4$) satisfies
\[
  \min\{\eff_\epsilon(p), n \} \leq \widehat{\bm S} \leq |\supp(p)| \,.
\]
In particular, there is a support size tester using $O(n/\epsilon)$ samples, obtained by
running this algorithm with parameters $n+1, \epsilon$ and outputting $\ACCEPT$ if and only if
$\widehat{\bm S} \leq n$.
\end{proposition}
\begin{proof}
Write $k \define \min\{n, \eff_\epsilon(p)\}$.  Let $m \define 10 n / \epsilon$. The algorithm draws
$m$ independent random samples from $p$ and outputs $\widehat{\bm S}$, defined as the number
of distinct domain elements in the sample.

It is clear that $\widehat{\bm S} \leq |\supp(p)|$ with probability 1. To show that $\widehat{\bm S}
\geq \min\{\eff_\epsilon(p), n\}$, suppose we draw an infinite sequence of independent random
samples $\bm S = \{ \bm{s}_1, \bm{s}_2, \dotsc \}$ from $p$. For each $i \in \bN$, define the random
variable $\bm{t}_i$ as the smallest index $\bm{t}_i \in [m]$ such that the multiset $\{ \bm{s}_1,
\dotsc, \bm{s}_{\bm{t}_i} \}$ is supported on $i$ unique elements. Note that $\bm{t}_1 =
1$.

\begin{claim}
For all $i \in \{2, \dotsc, k\}$, $\Ex{ \bm{t}_i - \bm{t}_{i-1} } \leq 1/\epsilon$.
\end{claim}
\begin{proof}[Proof of claim]
Write $\bm{S}_t \define \{ \bm{s}_1, \dotsc, \bm{s}_t \}$ as the prefix of $\bm{S}$ up to element
$t$.  Fix any $t_{i-1}$. Conditional on the event $\bm{t}_{i-1} = t_{i-1}$, the multiset
$\bm{S}_{t_{i-1}} = \{\bm{s}_1, \dotsc, \bm{s}_{t_{i-1}}\}$ is supported on $i-1 < k$ elements. Then
$p( \supp(\bm{S}_{t_{i-1} }) ) < 1-\epsilon$ since $p$ is $\epsilon$-far from having support size at
most $k-1$ by \cref{obs:eff}.
Therefore for every $t > t_{i-1}$ we have $\Pruc{}{ \bm{s}_t \not\in \bm{S}_{t_{i-1}}
}{\bm{t}_{i-1} = t_{i-1}} > \epsilon$, so $\Ex{ \bm{t}_i - \bm{t}_{i-1} } \leq 1/\epsilon$.
\end{proof}
From this claim we may deduce $\Ex{\bm{t}_k} \leq \bm{t}_1 + (k-1)/\epsilon \leq n/\epsilon$,
so by Markov's inequality
\[
  \Pr{\bm{t}_k > 10n/\epsilon} \leq 1/10 \,.
\]
Therefore if $m \geq 10n/\epsilon$ the tester will see at least $k$ unique elements in a sample of
size $m$, with probability at least $9/10$.
\end{proof}

\subsection{The test statistic}

We now turn to the case where $\epsilon$ is large. Specifically, we fix a small universal constant
$a \in (0,1)$ (\eg $a = 1/128$ suffices; we prioritize clarity of exposition over optimizing
constants) and assume:

\begin{assumption}
\label{assumption:epsilon}
For every $n \in \bN$, we assume $n^{-a} < \epsilon < 1/3$,
so that $\log(1/\epsilon) < a \log n$.
\end{assumption}

We will write $\bm{S} \subset \bN$ for the (random) multiset of samples received by the algorithm.

\begin{definition}[Sample Histogram]
For a fixed multiset $S \subset \bN$, the \emph{sample histogram} is the sequence $N_i$ where $N_i$
is the number of times $i \in \bN$ appears in $S$. If $\bm S$ is a random multiset then we write
$\bm{N}_i$.
\end{definition}

The goal will be to find the best function $f$ to plug in to the following definition of our test
statistic; we will show how to choose $f$ in \cref{section:choosing_f} and the remainder of the
paper, but for now we leave it undetermined.

\begin{definition}[Test statistic]
\label{def:test-statistic}
For a given function $f : \{0\} \cup \bN \to \bR$ (which is required to satisfy $f(0) = -1$),
we define a test statistic as follows. On random sample $\bm S$,
\[
    \bm{\widehat S} \define \sum_{i \in \bN} (1 + f(\bm{N_i}) ) \,.
\]
\end{definition}

Given an appropriate function $f$ and resulting test statistic, we define our support-size tester:

\begin{definition}[Support-Size Tester]
\label{def:tester}
Given parameters $n \in \bN$ and $\epsilon \in (0,1)$ our tester chooses $m$ independent samples
$\bm{S}$ from the distribution $p$, and outputs \textsf{Accept} if and only if
\[
  \bm{\widehat S} < (1 + \epsilon/2) n \,.
\]
\end{definition}

\begin{remark}
Our test statistic is a \emph{linear estimator} (see \eg \cite{VV11stoc,VV11focs,VV17,WY19}) because
it can be written as
\[
  \widehat{\bm S} = \sum_{j=1}^m \bm{F_j} \cdot (1+f(j)) \,,
\]
where $\bm{F_j}$ is the number of elements $i \in \bN$ which occur $j$ times in the sample $\bm{S}$,
\ie $\bm{N_i} = j$. The sequence $\bm{F_1}, \bm{F_2}, \dotsc$ is known as the \emph{fingerprint}.
\end{remark}

A standard trick to ease the analysis is to consider instead the \emph{Poissonized} version of the
algorithm.
\begin{definition}[Poissonized Support-Size Tester]
\label{def:poisson-tester}
Given parameters $n \in \bN$ and $\epsilon \in (0,1)$ our tester chooses $\bm{m} \sim \Poi(m)$ and
then takes $\bm{m}$ independent samples $\bm{S}$ from the distribution $p$, and outputs
\textsf{Accept} if and only if
\[
  \bm{\widehat S} < (1 + \epsilon/2) n \,.
\]
\end{definition}

The advantages of Poissonization are these (see \eg the survey \cite{Can20}):

\begin{fact}
\label{fact:poisson}
Let $\bm S$ be the sample of size $\bm m \sim \Poi(m)$ drawn by the Poissonized tester.
Then the random variables $\bm{N_i}$ are mutually independent and are distributed as
\[
  \bm{N_i} \sim \Poi(m p_i) \,.
\]
\end{fact}

\begin{fact}
\label{fact:poissonization}
If there is a Poissonized support-size tester with sample complexity $m(n,\epsilon)$ and success
probability $\successprob$, then there is a standard support-size tester with sample complexity at
most $O(m(n,\epsilon))$ and success probability $0.99 \successprob$.
\end{fact}

\subsection{Choosing a function \texorpdfstring{$f$}{f}}
\label{section:choosing_f}

Now we see how to choose the function $f$ to complete \cref{def:test-statistic}. To motivate the
choice, compute the expected value of the statistic:

\begin{proposition}
For any given finitely supported $f \colon \{0\} \cup \bN \to \bR$ with $f(0) = -1$ and parameter
$m$, the (Poissonized) test statistic $\widehat{\bm S}$ satisfies
\begin{equation}
\label{eq:mean-P}
  \Ex{ \widehat{ \bm S } }
  = \sum_{i \in \bN} ( 1 + e^{-mp_i} P(p_i) )
\end{equation}
where $P(x)$ is the polynomial
\[
  P(x) = \sum_{j=0}^\infty f(j) \frac{m^j}{j!} x^j \,.
\]
\end{proposition}
\begin{proof}
Using \cref{fact:poisson},
\begin{align*}
        \Ex{\bm{\widehat S}}
        &= \sum_{i \in \bN} \Ex{1 + f(\bm{N_i}) }
         = \sum_{i \in \bN} \left( 1 + \sum_{j \ge 0} \Pr{\bm{N}_i = j} f(j) \right) 
        = \sum_{i} \left( 1 + \sum_{j \ge 0} \frac{e^{-mp_i}(mp_i)^j}{j!} f(j) \right) \\
        &= \sum_{i} \left(
            1 + e^{-mp_i} \sum_{j \ge 0} \frac{m^j f(j)}{j!} p_i^j \right)
        = \sum_{i} \left( 1 + e^{-mp_i} P(p_i) \right). \qedhere
\end{align*}
\end{proof}
For convenience, we define
\begin{definition}
\label{def:Q}
For given polynomial $P(x)$, define
\[
  Q(x) \define 1 + e^{-mx} P(x) \,,
\]
\end{definition}
\noindent
so that $Q(p_i)$ quantifies the contribution of element $i \in \bN$
to the expected value of $\bm{\widehat S}$:
\begin{equation}
\label{eq:mean-Q}
  \Ex{\bm{\widehat S}} = \sum_{i \in \bN} Q(p_i) \,.
\end{equation}
If we could set $Q(x) = \ind{x > 0}$ then $\Ex{\bm{\widehat S}} = |\supp(p)|$.  But this
means $P(0) = -1$ while $P(x) = 0$ for all $x > 0$, which is not a polynomial.
Our goal is therefore to choose a polynomial $P$ which approximates this function as closely as
possible.

Specifically, if we choose any coefficients $a_1, \dotsc, a_d$, we may then define
\[
    f(j) \define \begin{cases}
      a_j \frac{j!}{m^j} &\text{ if } j \in [d] \\
      -1 &\text{ if } j = 0 \\
      0 &\text{ if } j > d 
    \end{cases}
\]
to obtain the degree-$d$ polynomial
\[
  P_d(x) \define \sum_{j=0}^\infty f(j) \frac{m^j}{j!} x^j = \left(\sum_{j=1}^d a_j x^j\right) - 1
  \,,
\]
which satisfies the necessary condition $P_d(0) = -1$. The other condition that we want $P_d(x)$ to
satisfy is that $P_d(x) \approx 0$ for $x > 0$. This cannot be achieved by a low-degree polynomial,
but we can instead ask for $P_d(x) \approx 0$ inside a chosen ``safe interval'' $[\ell, r]$. In other words,
we want a polynomial which satisfies $P_d(0) = -1$ and
\[
  \max_{x \in [\ell, r]} |P_d(x)| \leq \delta \,,
\]
for some small $\delta$. The lowest-degree polynomials which satisfy these conditions are known
as the (shifted and scaled) \emph{Chebyshev polynomials}, which we define in
\cref{section:chebyshev_polynomials}.

\begin{remark}
In the support-size \emph{estimation} problem, there is a lower bound $p_i \geq 1/n$ on the
nonzero densities, so \cite{WY19} choose a ``safe interval'' $[\ell, r]$ with $\ell = 1/n$,
so that there is no density $p_i$ to the left of the interval, $p_i \in (0,\ell)$.
In our problem, we do not have this guarantee, and therefore we must handle values $p_i$ that do not
fall in the ``safe'' interval $[\ell, r]$.
\end{remark}

\subsection{Chebyshev polynomials and the definition of \texorpdfstring{$P_d$}{Pd}}
\label{section:chebyshev_polynomials}

For a given degree $d$, the Chebyshev polynomial $T_d(x)$ is the polynomial which grows fastest on
$x > 1$ while satisfying $|T_d(x)| \leq 1$ in the interval $x \in [-1,1]$.

\begin{definition}[Chebyshev polynomials]
We write $T_d(x)$ for the degree-$d$ Chebyshev polynomial. These can be defined recursively as
\begin{align*}
  T_0(x) &\define 1 \\
  T_1(x) &\define x \\
  T_d(x) &\define 2x \cdot T_{d-1}(x) - T_{d-2}(x) \,.
\end{align*}
$T_d(x)$ may also be computed via the following closed-form formula:
    \label{fact:chebyshev-formula}
    For all $d \in \bN$ and $x \in \bR$, 
    \[
        T_d(x)
        = \frac{1}{2} \left( \left(x + \sqrt{x^2-1}\right)^d
            + \left(x - \sqrt{x^2-1}\right)^d \right) \,.
    \]
\end{definition}

We want the polynomials $P_d(x)$ to be bounded in $[\ell, r]$ and to 
grow fast on $x \in [0,\ell)$ so that $P_d(0) = -1$. So we define a map $\psi$ which translates
$[\ell, r]$ to $[-1,1]$,
\[
  \psi(x) \define -\frac{2x - r - \ell}{r-\ell}
\]
so $\psi(0) = \frac{r+\ell}{r-\ell}$. For convenience we write
\[
  \alpha \define \frac{\ell}{r}
\]
so that $\psi(0) = 1 + \frac{2\alpha}{1-\alpha}$. We now define

\begin{definition}[Shifted \& Scaled Chebyshev Polynomials]
For any $d$, define
\[
  P_d(x) \define - \delta T_d(\psi(x)) \,,
\]
where $\delta \define \frac{1}{T_d\left(1 + \frac{2\alpha}{1-\alpha}\right)}$.
\end{definition}
This way, as required,
\[
  P_d(0) = - \delta T_d(\psi(0)) = - 1 \,, \;\text{ and }\; |P_d(x)| \leq \delta \text{ for } x \in [\ell,
r] \,.
\]
We illustrate the Chebyshev polynomial $T_d(x)$ and the shifted and scaled result $Q(p_i)$
in \cref{fig:cheb-example}. The resulting function values $f(\bm{N_i})$ which define our test
statistic are given in \cref{res:f-values}.

\begin{figure}[t]
    \centering
    \begin{subfigure}{0.46\textwidth}
        \centering
        \includegraphics[width=\textwidth]{"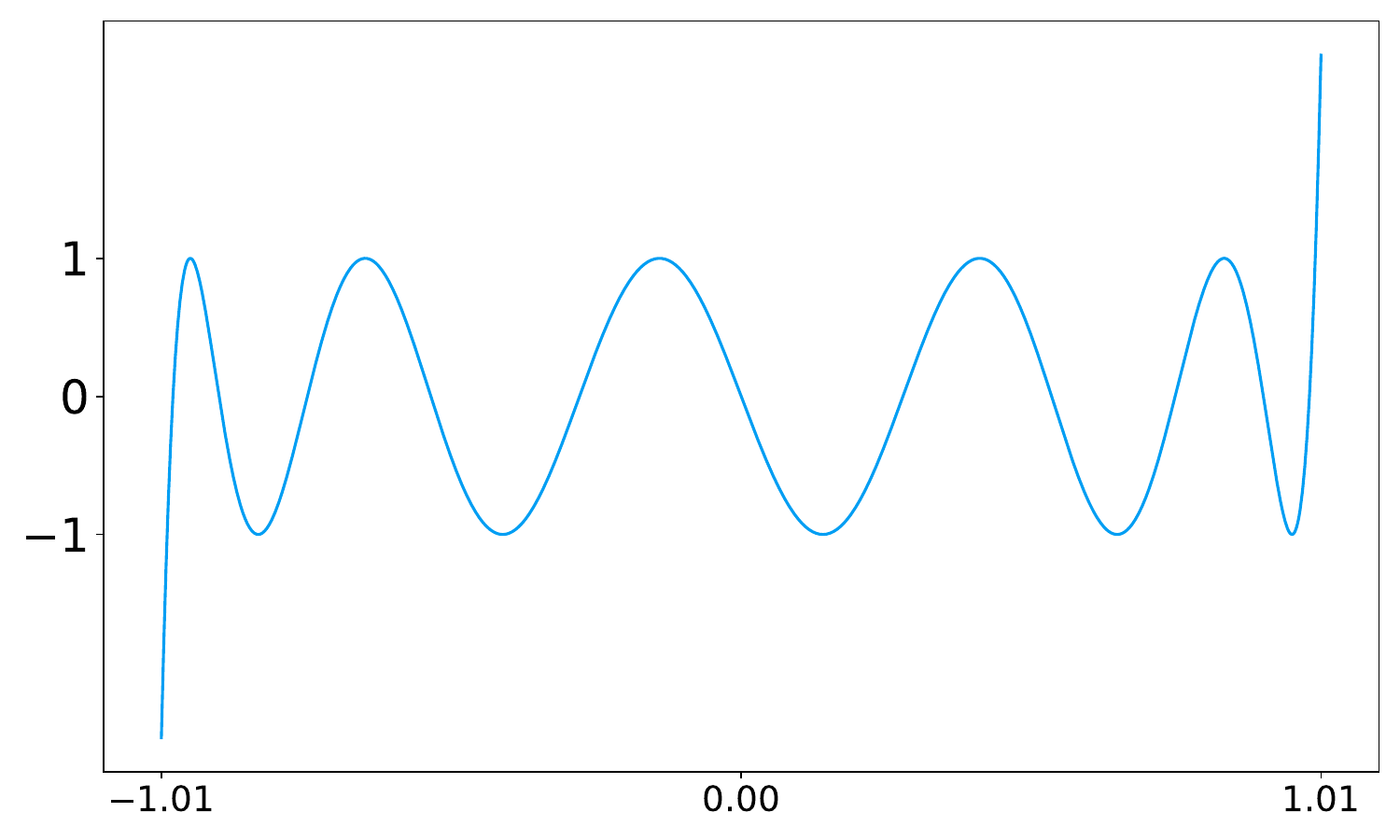"}
        \caption{$T_{11}(x), x \in [-1.01, 1.01]$.}
        \label{fig:cheb-example-a}
    \end{subfigure}
    \hspace{0.05\textwidth} 
    \begin{subfigure}{0.46\textwidth}
        \centering
        \includegraphics[width=\textwidth]{"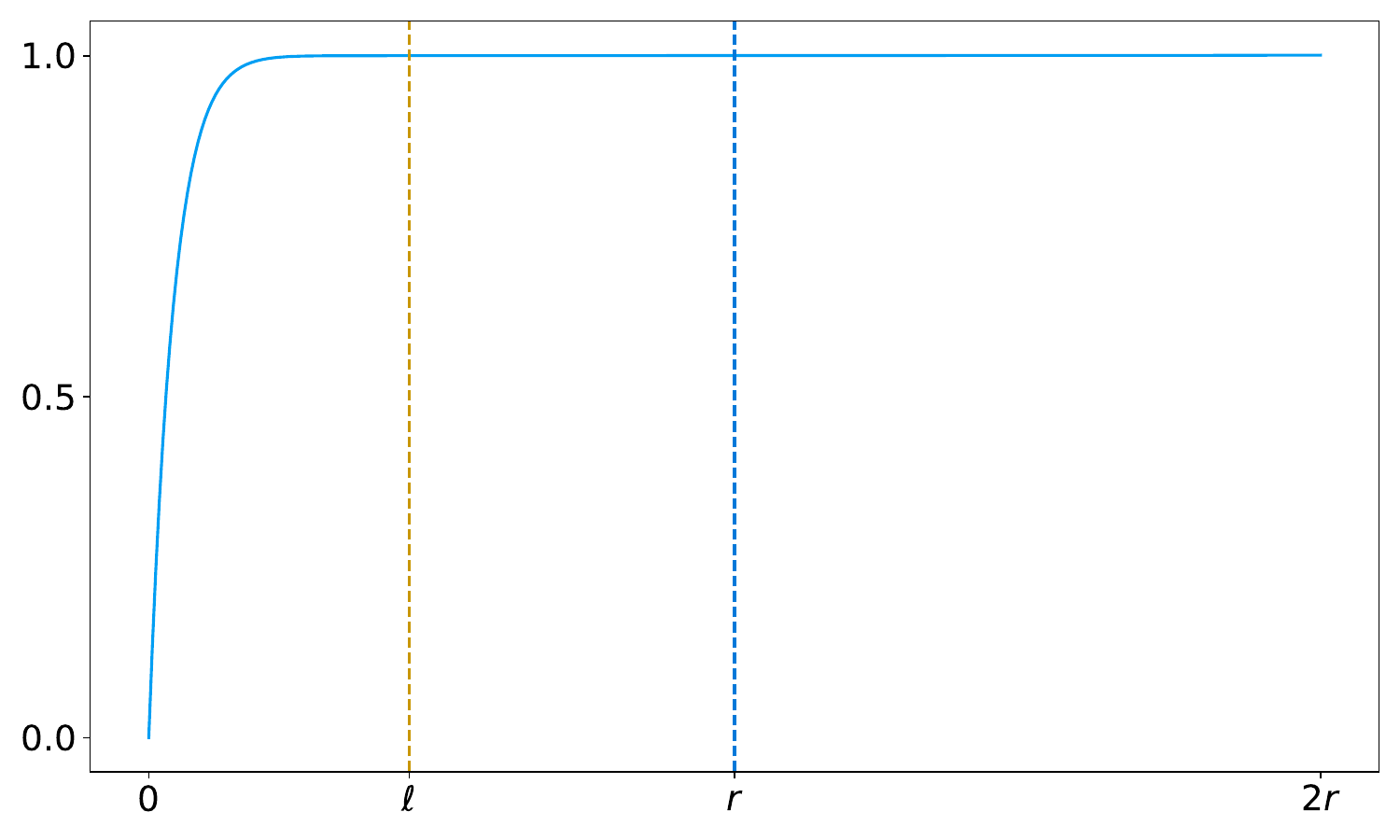"}
        \caption{$Q(p_i) = 1 + e^{-mp_i} P_{11}(p_i)$.}
        \label{fig:cheb-example-b}
    \end{subfigure}
    \caption{The polynomial $T_{11}(x)$ and the resulting $Q(p_i) = 1 + e^{-mp_i} P_d(p_i)$
    for $d=11$ and certain choices of $\ell, r, n, m$. The `safe' interval $[\ell, r]$ is between
the two vertical lines in \cref{fig:cheb-example-b}. See that $Q(p_i)$ is an approximation of the
`idealized' function $Q(p_i) = \ind{p_i > 0}$.}
\label{fig:cheb-example}
\end{figure}

\newpage

\section{Choosing Parameters}
\label{section:parameters}

Our goal is now:

\begin{goal}
\label{goal:main}
Choose degree $d$ and ``safe interval'' $[\ell, r]$ such that the (Poissonized) support-size tester
in \cref{def:poisson-tester}, instantiated with the (shifted and scaled) Chebyshev polynomial $P_d$,
works with the smallest sample size parameter $m$, even though densities $p_i$ may not belong to
$[\ell, r]$.
\end{goal}

To do this, we will set up a series of constraints that these parameters must satisfy in order for
the tester to succeed, and then minimize $m$ with respect to these parameters.

\subsection{Constraint I: Choose \texorpdfstring{$d$}{d} such that \texorpdfstring{$\delta \leq \epsilon$}{delta at most epsilon}}

We will want that $\delta \leq \epsilon$ so that for every density $p_i$ in the ``safe interval''
$[\ell, r]$ we have $Q(p_i) \approx 1$, specifically $|Q(p_i) - 1| \leq \epsilon$. We will write this
constraint in terms of $\ell$ and $r$:

\begin{constraint}
\label{constraint:d_log}
For constant $C_d \define 4\ln(2)$, we will require for any choice of $\ell, r, \epsilon$
that $r \ge 3\ell$ and that the degree $d$ satisfies
\[
  d \geq C_d \sqrt{\frac{r-\ell}{2\ell}} \log \left(\frac{20}{\epsilon} \right) \,.
\]
\end{constraint}

\begin{proposition}
\label{res:d_log}
Assume \cref{constraint:d_log}. Then $\delta \leq \epsilon/20$ and for all $x \in [\ell, r]$,
\[
    | Q(x) - 1 | \leq \delta \le \frac{\epsilon}{20} \,.
\]
\end{proposition}

To prove this, we require a lower bound on the growth rate of the Chebyshev polynomials. This is well
known but we provide a proof for the sake of completeness:

\begin{proposition}
\label{res:cheb-tail-lb}
There is a universal constant $c \define \frac{1}{2\ln(2)}$ such that, for all $d \in \bN$ and
$\gamma \in [0,1]$,
\[
  T_d(1 + \gamma) \geq 2^{c \cdot d \sqrt \gamma - 1} \,.
\]
\end{proposition}
\begin{proof}
    Using \cref{fact:chebyshev-formula}, we have
    \begin{align*}
        T_d(1+\gamma)
        &= \frac{1}{2} \left( \left((1+\gamma) + \sqrt{(1+\gamma)^2-1}\right)^d
            + \left((1+\gamma) - \sqrt{(1+\gamma)^2-1}\right)^d \right) \\
        &\ge \frac{1}{2} \left( 1 + \sqrt{2\gamma} \right)^d
        = \frac{1}{2} e^{d \ln(1 + \sqrt{2\gamma})} \,.
    \end{align*}
    Using the inequality $\ln(1+x) \ge \frac{x}{1+x}$, which is valid for $x > -1$, together with
    the assumption that $\gamma \le 1$, we obtain
    \[
        T_d(1+\gamma)
        \ge \frac{1}{2} e^{d \cdot \frac{\sqrt{2\gamma}}{1 + \sqrt{2\gamma}}}
        \ge \frac{1}{2} e^{d \cdot \frac{\sqrt{2\gamma}}{2\sqrt{2}}}
        = \frac{1}{2} e^{\frac{1}{2} d \sqrt{\gamma}} \,. \qedhere
    \]
\end{proof}

\begin{proof}[Proof of Proposition~\ref{res:d_log}]
Recall $\alpha \define \ell/r$ and $\psi(0) = 1 + \tfrac{2\alpha}{1-\alpha}$.
It suffices to show
\[
  T_d\left(1 + \frac{2\alpha}{1-\alpha}\right) \geq \frac{20}{\epsilon}
\]
so that, for $x \in [\ell, r]$ (which satisfies $\psi(x) \in [-1,1]$ and therefore $T_d(\psi(x)) \in
[-1,1]$),
\[
  P_d(x) = - \delta T_d(\psi(x)) = - \frac{T_d(\psi(x))}{T_d(\psi(0))}
\]
is within the interval $[-\epsilon/20, \epsilon/20]$.
By \cref{constraint:d_log}, $\tfrac{2\alpha}{1-\alpha} \le 1$ since $\alpha \le 1/3$, so
we may apply \cref{res:cheb-tail-lb} to get
\[
    T_d\left(1 + \frac{2\alpha}{1-\alpha}\right)
    \geq 2^{c \cdot d \sqrt{\frac{2\alpha}{1-\alpha}} - 1} \,,
\]
where $c = \frac{1}{2\ln(2)}$ is the constant from \cref{res:cheb-tail-lb}. Therefore when
$d
\geq \frac{2}{c} \sqrt{\frac{r-\ell}{2\ell}} \log(20/\epsilon)
= \frac{2}{c} \sqrt{\frac{1-\alpha}{2\alpha}} \log(20/\epsilon)
\geq \frac{1}{c} \sqrt{\frac{1-\alpha}{2\alpha}} (1 + \log(20/\epsilon))$ we have the result
with $C_d = \frac{2}{c} = 4\ln(2)$.
\end{proof}

\subsection{Constraint II: Choose \texorpdfstring{$m$}{m} such that \texorpdfstring{$Q(p_i) \approx 1$ when $p_i > r$}{Q(pi) approximately 1 when p i > r}}

Next we need to choose large enough $m$ such that the term $e^{-mp_i}$ cancels out the growth of
$P_d(p_i)$ to the right of the ``safe interval''\!, leading to the desired $Q(p_i) = 1 + e^{-mp_i}
P_d(p_i) \approx 1$ in this case. (\cref{fig:constraint_2} shows $1 + P_d(p_i)$ without the
$e^{-mp_i}$ term.)

\begin{figure}
\centering
\includegraphics[width=0.48\textwidth]{"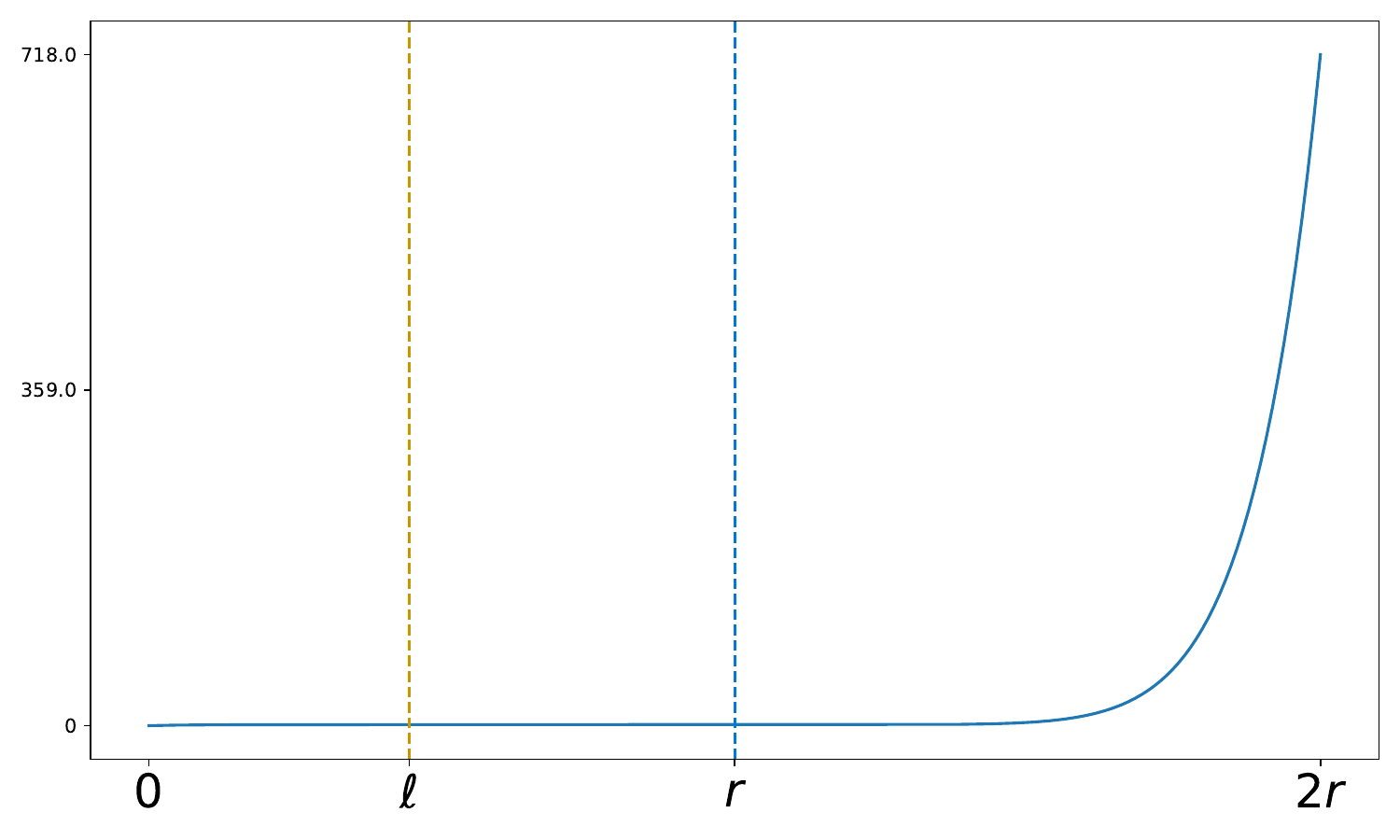"}
\caption{Untamed right tail.}
\label{fig:constraint_2}
\end{figure}%
We place the following constraint:

\begin{constraint}
\label{constraint:right_tail}
We require $m \ge 5.5 \cdot d/(r-\ell)$.
\end{constraint}

\begin{proposition}
\label{res:right_tail}
Assume \cref{constraint:right_tail}. Then 
\[
  \forall x \in (r, 1] \;\colon\; |1 - Q(x)| \leq \delta \,.
\]
\end{proposition}
\begin{proof}
Write $\gamma = 2\frac{x-r}{r-\ell}$ so that
\[
  - \psi(x) =  \frac{2x-r-\ell}{r-\ell}
            = \frac{r-\ell - 2r + 2x}{r-\ell}
            = 1 + \frac{2x-2r}{r-\ell}
            = 1 + \gamma \,.
\]
Note that $|T_d(z)|=|T_d(-z)|$ for all $z$.  We need to bound $|1-Q(x)| = e^{-mx} |P_d(x)| = \delta e^{-mx}
| T_d(\psi(x)) |$ for $x > r$, and $|T_d(\psi(x))| = T_d(-\psi(x)) = T_d(1+\gamma)$.  So we require
$e^{mx} \geq T_d(1 + \gamma)$. Bound $T_d(1+\gamma)$ by
\begin{align*}
T_d(1+\gamma)
&= \frac{1}{2}\left((1 + \gamma + \sqrt{2\gamma + \gamma^2})^d + (1+\gamma - \sqrt{2\gamma +
\gamma^2})^d\right) \\
&\leq (1+\gamma+\sqrt{2\gamma+\gamma^2})^d
\leq e^{d(\gamma+\sqrt{2\gamma+\gamma^2})} \,.
\end{align*}
If $\gamma < 1$ then $\gamma + \sqrt{2\gamma+\gamma^2} < 1 + \sqrt{3}$, and therefore since
$\frac{x}{r-\ell} > \frac{x}{r} > 1$ we have
\[
  d(\gamma + \sqrt{2\gamma + \gamma^2}) \leq d (1+\sqrt{3}) \leq d(1+\sqrt{3}) \frac{x}{r-\ell}
  \leq mx \,.
\]
If $\gamma \geq 1$ then $\gamma + \sqrt{2\gamma+\gamma^2} \leq (1+\sqrt{3})\gamma$, so
\[
  d(\gamma + \sqrt{2\gamma + \gamma^2}) \leq d (1+\sqrt{3}) \cdot 2\frac{x-r}{r-\ell}
    \leq d (1+\sqrt{3}) \cdot 2\frac{x}{r-\ell} \leq mx \,.
\]
We then conclude
\[
  T_d(1+\gamma) \leq e^{d(\gamma + \sqrt{2\gamma + \gamma^2})} \leq e^{mx}. \qedhere
\]
\end{proof}

\subsection{Constraint III: Variance}

We will need the test statistic $\widehat{\bm S}$ to satisfy two conditions.
First, when $|\supp(p)| \leq n$ we need
\[
  \Pr{ \widehat{\bm S} > (1+\epsilon/2)n } \leq 1/4 \,.
\]
Second, when $p$ is $\epsilon$-far from having support size $n$, we need
\[
  \Pr{ \widehat{\bm S} < (1+\epsilon/2)n } \leq 1/4 \,.
\]
We will bound $\Ex{\widehat{\bm S}}$ in each of these cases in \cref{section:means}. But first we
will bound the variance.  We impose the following constraint:

\begin{constraint}
\label{constraint:variance}
We require $m \leq \tfrac{1}{4^4} \epsilon^2 n^2$,
and $d^6 9^d \left(\frac{r+\ell}{r-\ell}\right)^{2d-2} \leq \tfrac{1}{4} m(r-\ell)^2 n^2$.
\end{constraint}

\begin{proposition}
\label{res:variance}
Assume \cref{constraint:d_log,constraint:right_tail,constraint:variance}. Then
$\Var{\widehat{\bm S}} \leq \epsilon^2 \frac{n^2}{4^3}$,
and therefore (by Chebyshev's inequality),
\[
  \Pr{ |\widehat{\bm S} - \Ex{\widehat{\bm S}}| > \frac{\epsilon}{4} n }
  \leq \frac{1}{4} \,.
\]
\end{proposition}

To prove \cref{res:variance}, we will require bounds on the values of $f(N_i)$ in the definition of
\[
  \widehat{\bm S} = \sum_{i \in \bN} ( 1 + f(\bm{N_i}) ) \,.
\]
For this we do some tedious calculations using the coefficients of the polynomial $P_d(x)$. We put
these in \cref{section:coefficients}. The result is:
\begin{restatable}{proposition}{resmaxu}
\label{res:max_u}\RestateRemark
For any $d$ and any $k \in [d]$,
\[
  |f(k)| \leq \delta \cdot d^2 \cdot 3^d \cdot \left( \frac{2d}{m (r-\ell)} \right)^k
         \left(\frac{r+\ell}{r-\ell}\right)^{d-k} \,.
\]
\end{restatable}

We may now estimate the variance.

\begin{proof}[Proof of Proposition~\ref{res:variance}]
Recall that each $\bm{N}_i$ is an independent Poisson random variable $\bm{N}_i \sim \Poi(mp_i)$.
Then
\begin{align*}
\Var{\bm{\widehat S}}
&= \sum_{i \in \bN} \Var{1+f(\bm{N_i})}
= \sum_{i \in \bN} \Var{(1+f(\bm{N_i})) \ind{\bm{N}_i > 0}} \\
&\leq \sum_{i \in \bN} \Ex{(1+ f(\bm{N_i}) )^2 \ind{\bm{N}_i > 0}} 
\leq \max_{j \in [d]} (1+ f(j) )^2 \cdot  \sum_{i \in \bN} \Ex{\ind{\bm{N}_i > 0}} \\
&= \max_{j \in [d]} (1+f(j) )^2 \cdot \sum_{i \in \bN} (1 - e^{-mp_i}) 
\leq \max_{j \in [d]} (1+ f(j) )^2 \cdot m \sum_{i \in \bN} p_i
= m \cdot \max_{j \in [d]} (1+ f(j) )^2 \,.
\end{align*}
From \cref{constraint:right_tail}, which implies $m(r-\ell) > 2d$, the upper bound in
\cref{res:max_u} is maximized at $k=1$, so we have
\[
\forall j \in [d] \;:\; |f(j)| \leq \delta \cdot d^2 \cdot 3^d \cdot \frac{2d}{m(r-\ell)} \cdot
\left(\frac{r+\ell}{r-\ell}\right)^{d-1} \,.
\]
If this is at most 1, then the first part of \cref{constraint:variance} gives the desired bound.
\[
  \Var{\widehat{\bm S}} \leq 4m \leq \frac{\epsilon^2 n^2}{4^3}
\]
Otherwise, $(1+|f(j)|)^2 \leq 4f(j)^2$, so (using \cref{constraint:d_log} and \cref{res:d_log} to
ensure $\delta \leq \epsilon/20 \le \epsilon/4^2$), we get from \cref{constraint:variance} that
\begin{align*}
  \Var{\widehat{\bm S}}
  &\leq 4^2 m \delta^2 \cdot d^6 \cdot 9^d \cdot \frac{1}{m^2(r-\ell)^2} \left( \frac{r+\ell}{r-\ell}
\right)^{2d-2}  \\
  &\leq \frac{1}{4^2} \epsilon^2 \cdot d^6 \cdot 9^d \cdot \frac{1}{m(r-\ell)^2} \left( \frac{r+\ell}{r-\ell}
\right)^{2d-2} 
  \leq \epsilon^2 \frac{n^2}{4^3} \,. \qedhere
\end{align*}
\end{proof}

\subsection{Constraint IV: \texorpdfstring{$\mathbb{E}[\widehat{S}] \geq (1+\epsilon)n$ when $p$ is $\epsilon$-far}{expectation of S hat large when p is epsilon-far}}
\label{section:means}

The above \cref{constraint:d_log,constraint:right_tail} ensure that $Q(x) \in 1 \pm \delta$ when $x
\in [\ell, 1]$. The properties of the Chebyshev polynomial also ensure that $Q(x) \leq 1+\delta$ for
$x < \ell$ (see the behaviour of $Q(x)$ for $x \leq \ell$ in \cref{fig:Q_dangerzone}). This means
$\Ex{\widehat{\bm S}} = \sum_{i : p_i > 0} Q(p_i)$ is essentially an underestimate of $|\supp(p)|$, which is
enough to guarantee that the tester will accept $p$ when $|\supp(p)| \leq n$ with high
probability:
\begin{lemma}
  \label{res:completeness}
  Assume \cref{constraint:d_log,constraint:right_tail}. Then
  \[
      \Ex{\bm{\widehat S}} \le (1 + \delta) |\supp(p)|
      < (1 + \epsilon/4) |\supp(p)| \,.
  \]
\end{lemma}
\begin{proof}
  First note that $\psi(x) > 1$ for all $x \in (0, \ell)$, so that $P_d(x) = -\delta
  T_d(\psi(x)) < 0$. Hence $Q(x) = 1 + e^{-mx} P_d(x) < 1$. Combining with
  \cref{res:d_log,res:right_tail} we conclude that, for all $x \in (0, 1]$, $Q(x) \le 1 + \delta$.
  We also have $Q(0) = 0$. Recalling \eqref{eq:mean-Q}, we obtain
  \[
      \Ex{\bm{\widehat S}}
      = \sum_{i \in \bN} Q(p_i)
      \le \sum_{i \in \bN : p_i > 0} (1 + \delta)
      = (1 + \delta) |\supp(p)|
      < (1 + \epsilon/4) |\supp(p)| \,,
  \]
  the last inequality since $\delta < \epsilon/4$ by \cref{res:d_log}.
\end{proof}

\begin{remark}
\label{remark:recover-wy}
If we were guaranteed that all nonzero densities satisfied $p_i \geq \ell$, then
\cref{res:completeness} would also show $\Ex{\widehat{\bm S}} \geq (1-\delta)|\supp(p)|$. In
particular, for the support-size \emph{estimation} task, where the nonzero densities satisfy $p_i >
1/n$, we could simply set $\ell = 1/n$, skip the remainder of this section, and recover the optimal
$O\left(\frac{n}{\log n} \log^2(1/\epsilon)\right)$ bound of \cite{WY19}. 
\end{remark}
Now we need to ensure that $\Ex{\widehat{\bm S}}$ will be large when $p$ is $\epsilon$-far from
having support size $n$. The difficulty is that the densities $p_i$ may be arbitrarily small and lie
outside the ``safe interval'' $[\ell,r]$ where the polynomial approximation guarantees $Q(p_i)
\approx 1$. So we will require properties of the Chebyshev polynomial approximation outside the
interval $[\ell, r]$. We impose the following constraint:

\begin{constraint}
\label{constraint:ell}
We require $\ell \le C_\ell \frac{\epsilon}{n}$, where $C_\ell$ is any constant $C_\ell \leq
1/20$.
\end{constraint}
The constraint will give us the desired guarantee:

\begin{lemma}
    \label{res:soundness}
    Assume \cref{constraint:d_log,constraint:right_tail,constraint:ell}. Then for every $k \leq n$,
    if $p$ is
    $\epsilon$-far from being supported on $k$ elements,
    \[
        \Ex{\bm{\widehat S}} > (1 + 3\epsilon/4) k \,.
    \]
\end{lemma}

\begin{remark}
\label{remark:better-constraint}
\cref{constraint:ell} and \cref{res:soundness} lead to an upper bound of $O\left(\frac{n}{\epsilon
\log n} \log^2(1/\epsilon)\right)$ instead of $O\left(\frac{n}{\epsilon \log n}
\log(1/\epsilon)\right)$.  We put them here to ease the exposition, since the argument is simpler.
The argument for the better bound (using a weaker constraint $\ell \leq C_\ell \tfrac{\epsilon}{n}
\log(1/\epsilon)$) is in \cref{section:improved-bound}.
\end{remark}

\begin{figure} 
\centering
\vspace{-2em}
\includegraphics[width=0.48\textwidth]{"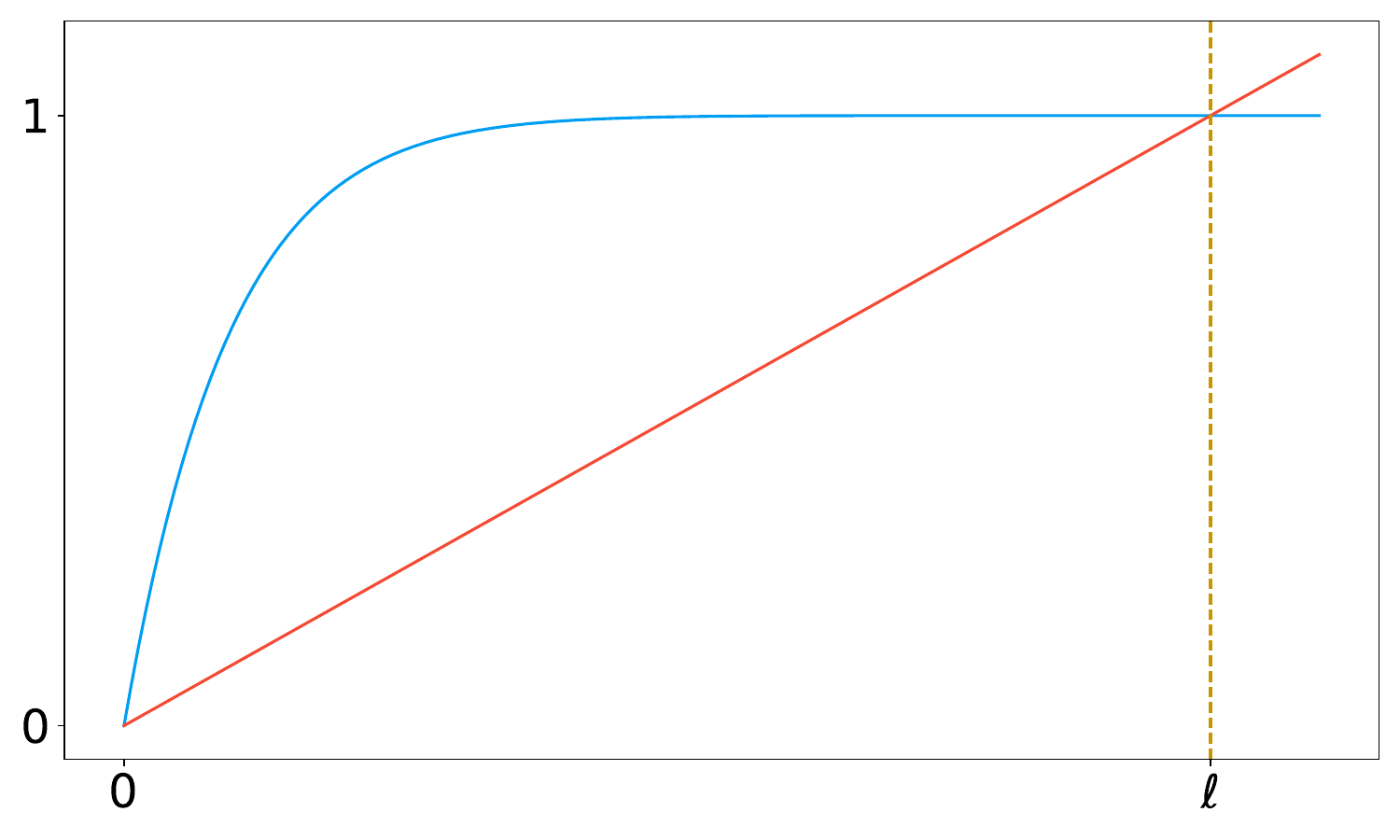"}
\caption{$Q(p_i)$ for $p_i < \ell$, and the linear lower bound $Q(p_i) \geq (1-\delta)\tfrac{p_i}{\ell}$
in \cref{res:q-light-bounds}.}
\label{fig:Q_dangerzone}
\end{figure}

To prove \cref{res:soundness}, we will use the concavity of $P_d(x)$ on $[0, \ell]$ to lower bound
$Q(p_i)$ by the line between the points $(0, 0)$ and $(\ell, Q(\ell))$. This is illustrated in
\cref{fig:Q_dangerzone}, which shows $Q$ compared to the linear lower bound. First, we need a
standard fact about the roots of Chebyshev polynomials.

\begin{fact}
    \label{fact:chebyshev-roots}
    For each $d \in \bN$, $T_d(x)$ has $d$ distinct roots, all within $[-1, 1]$.
\end{fact}

\begin{corollary}
    \label{res:chebyshev-convex-tail}
    For each $d \in \bN$, $T_d(x)$ is convex increasing in $[1, +\infty)$.
\end{corollary}
\begin{proof}
    Since $T_d(x)$ has a positive leading coefficient, all of its first $d$ derivatives are
    eventually positive nondecreasing as $x \to +\infty$. Moreover, all of its derivatives have
    their roots within $[-1, 1]$ by \cref{fact:chebyshev-roots} and the Gauss-Lucas theorem. Hence
    all of its first $d$ derivatives are positive on $(1, +\infty)$.
\end{proof}

\begin{corollary}
    \label{res:pd-concave}
    $P_d$ is concave increasing on $[0, \ell]$.
\end{corollary}
\begin{proof}
    This follows from the \cref{res:chebyshev-convex-tail} and the fact that $P_d = -\delta
    T_d(\psi(x))$, with $\psi$ an affine function mapping $[0, \ell]$ onto $\left[ 1, 1 +
    \frac{2\alpha}{1-\alpha} \right]$.
\end{proof}

\begin{proposition}
    \label{res:q-light-bounds}
    For each $x \in [0, \ell]$, we have $(1-\delta) \frac{x}{\ell} \le Q(x) \le 1$.
\end{proposition}
\begin{proof}
    Recall that $P_d(0) = -1$ by construction. It is also immediate to check that $T_d(1) = 1$ and
    hence $P_d(\ell) = -\delta T_d(\psi(\ell)) = -\delta T_d(1) = -\delta$. By
    \cref{res:pd-concave}, $P_d(x)$ is concave increasing (and hence negative) on $[0, \ell]$.
    Thus for any $x \in [0, \ell]$ we have
    \[
        Q(x) = 1 + e^{-mx} P_d(x)
        \le 1 + e^{-mx} \cdot 0
        = 1
    \]
    and, on the other hand,
    \[
        Q(x) = 1 + e^{-mx} P_d(x)
        \ge 1 + P_d\left(
            \frac{x}{\ell} \cdot \ell + \left(1-\frac{x}{\ell}\right) \cdot 0
        \right)
        \ge 1 + \left(
            \frac{x}{\ell} \cdot (-\delta) + \left(1-\frac{x}{\ell}\right) \cdot (-1)
        \right)
        = (1-\delta) \frac{x}{\ell} \,.
    \]
\end{proof}

We now calculate the lower bound on the test statistic by partitioning the densities $p_i$ into
weight classes.

\begin{definition}[Heavy and light elements]
    \label{def:heavy}
    An element $i \in \supp(p)$ is called \emph{heavy} if $p_i \ge \ell$;
    otherwise, if $p_i \in (0, \ell)$, it is called \emph{light}.
\end{definition}

\begin{lemma}[Refinement of \cref{res:soundness}]
    \label{res:soundness-refinement}
    Assume \cref{constraint:d_log,constraint:right_tail,constraint:ell}. For any distribution $p$,
    let $n_H$ be the number of heavy elements in $\supp(p)$, and let $\mu_L$ be the total
    probability mass of the light elements in $\supp(p)$. Then
    \begin{equation}
        \label{eq:test-statistic-lb}
        \Ex{\bm{\widehat S}} \ge (1-\delta) \left( n_H + \frac{\mu_L}{\ell} \right) \,.
    \end{equation}
    In particular, for any $k \leq n$, if $p$ is $\epsilon$-far from being supported on $k$
    elements, then
    \[
        \Ex{\bm{\widehat S}} > (1 + 3\epsilon/4) k \,.
    \]
\end{lemma}
\begin{proof}
    Let $H \subset \supp(p)$ denote the set of heavy elements, so that $n_H = |H|$. Recall that, by
    \eqref{eq:mean-Q},
    \[
        \Ex{\bm{\widehat S}}
        = \sum_{i \in \bN} Q(p_i)
        = \sum_{i \in H} Q(p_i) + \sum_{i \in \supp(p) \setminus H} Q(p_i) \,.
    \]
    By \cref{res:d_log,res:right_tail}, the first summation is
    \[
        \sum_{i \in H} Q(p_i)
        \ge \sum_{i \in H} (1-\delta)
        = (1-\delta) n_H \,.
    \]
    By \cref{res:q-light-bounds}, the second summation is
    \[
        \sum_{i \in \supp(p) \setminus H} Q(p_i)
        \ge \sum_{i \in \supp(p) \setminus H} (1-\delta) \frac{p_i}{\ell}
        = (1-\delta) \frac{\mu_L}{\ell} \,.
    \]
    This establishes \eqref{eq:test-statistic-lb}. Now, suppose $p$ is $\epsilon$-far from being
    supported on $k \le n$ elements. We consider two cases.

    \textbf{Case 1.} Suppose $\mu_L > \epsilon/10$. Then, using $\delta \le \epsilon/20$ (from
    \cref{res:d_log}) and $\ell \le \frac{\epsilon}{20n} \leq \frac{\epsilon}{20k}$
    (from \cref{constraint:ell}),
    \[
        \Ex{\bm{\widehat S}}
        \ge (1-\delta) \frac{\mu_L}{\ell}
        > (1 - \epsilon/20) \cdot \frac{\epsilon/10 \cdot 20k}{\epsilon}
        = (1 - \epsilon/20) \cdot 2k
        > (1 + 3\epsilon/4) k \,.
    \]

    \textbf{Case 2.} Suppose $\mu_L \le \epsilon/10$. Then there must exist a partition $H = R \cup
    T$ of the heavy elements such that $|R| = k$ and $p(T) > 9\epsilon/10$, since otherwise $p$ would
    be $\epsilon$-close to some distribution supported on at most $k$ elements. In fact, we can
    choose this partition such that $p(i) > p(j)$ for every $i \in R$ and $j \in T$. It follows that
    $p_j \le 1/k$ for every $j \in T$, since otherwise we would have $p(R) > \frac{1}{k} \cdot |R| =
    1$. Hence
    \[
        |T| \ge \frac{p(T)}{1/k} > \frac{9\epsilon/10}{1/k} = \frac{9 k \epsilon}{10} \,.
    \]
    Thus we have
    \begin{align*}
        \Ex{\bm{\widehat S}}
        &\ge (1-\delta) n_H
        \ge (1-\epsilon/20) \left( |R| + |T| \right)
        > (1-\epsilon/20) \left( k + \frac{9 k \epsilon}{10} \right) \\
        &= (1-\epsilon/20) \cdot (1 + 9\epsilon/10) k
        > (1 + 9\epsilon/10 - \epsilon/20 \cdot 2) k
        > (1 + 3\epsilon/4) k \,. \qedhere
    \end{align*}
\end{proof}

\subsection{Correctness of the tester and optimizing the parameters}
\label{section:correctness_and_parameters}

Let us repeat our list of constraints:
\begin{itemize}
\item \cref{constraint:d_log}: $r \ge 3\ell$ and
    $d \geq C_d \sqrt{\frac{r-\ell}{2\ell}} \log \left(\frac{20}{\epsilon} \right)$, where $C_d =
    4\ln(2)$.
\item \cref{constraint:right_tail}: $m \geq 5.5 d/(r-\ell)$.
\item \cref{constraint:variance}:
    $m \leq \tfrac{1}{4^4} \epsilon^2 n^2$ and
    $d^6 9^d \left(\frac{r+\ell}{r-\ell}\right)^{2d-2} \leq \tfrac{1}{4} m(r-\ell)^2 n^2$.
\item \cref{constraint:ell}: $\ell \leq C_\ell \frac{\epsilon}{n}$ for any $C_\ell \leq 1/20$.
\end{itemize}
We show in \cref{section:correctness} that if these constraints are satisfied, then the tester defined in \cref{def:tester}
will be correct. Then we optimize $m$ according to these constraints in \cref{section:optimizing}.

\subsubsection{Correctness}
\label{section:correctness}

\begin{lemma}[Correctness (Poissonized)]
\label{res:poisson-correctness}
Suppose $n, \epsilon$ satisfy \cref{assumption:epsilon}, and that
\cref{constraint:d_log,constraint:right_tail,constraint:variance,constraint:ell} are satisfied with
some sample size $m = m(n,\epsilon)$. Then the Poissonized test statistic $\widehat{\bm S}$
in \cref{def:test-statistic} satisfies
\[
  (1+3\epsilon/4) \min\{ \eff_\epsilon(p)-1, n \} < \Ex{ \widehat{\bm S}} < (1+\epsilon/4)|\supp(p)|
\]
and
\[
  \Pr{ |\widehat{\bm S} - \Ex{\widehat{\bm S}}| > \epsilon n/4} \le 1/4 \,.
\]
In particular, if $p$ is $\epsilon$-far from having support size at most $n$, then $\eff_\epsilon(p)
> n$ (\cref{obs:eff}),
so $\widehat{\bm S} > (1+\epsilon/2)n$ with probability at least $3/4$; and if $|\supp(p)|
\leq n$ then $\widehat{\bm S} < (1+\epsilon/2)n$ with probability at least $3/4$. So the algorithm
in \cref{def:poisson-tester} is a correct (Poissonized) support-size tester with sample complexity
$m$.
\end{lemma}
\begin{proof}
    \cref{res:completeness} gives that $\Ex{\bm{\widehat S}} < (1 + \epsilon/4) |\supp(p)|$, while
\cref{res:soundness} gives that $\Ex{\bm{\widehat S}} > (1 + 3\epsilon/4) \min\{\eff_\epsilon(p)-1,
n\}$.  Combining these with Chebyshev's inequality via \cref{res:variance} gives the result.
\end{proof}

Translating to the standard (non-Poissonized) testing model using \cref{fact:poissonization}:

\begin{corollary}[Correctness (Standard)]
\label{res:correctness}
Suppose that, for $n, \epsilon$ satisfying \cref{assumption:epsilon},
\cref{constraint:d_log,constraint:right_tail,constraint:variance,constraint:ell} are satisfied with
some sample size $m = m(n,\epsilon)$. Then (also for $n, \epsilon$ satisfying
\cref{assumption:epsilon}), there is a support-size tester with sample complexity $O( m(n,\epsilon)
)$.
\end{corollary}

\subsubsection{Optimizing the parameters}
\label{section:optimizing}

Now we complete \cref{goal:main} and find the setting of parameters which minimizes $m$ while
satisfying all of our constraints. The (weaker version of the) main \cref{res:intro-main} follows
from \cref{res:optimization}, using \cref{res:correctness}, and \cref{res:naive-estimator} to handle
small $\epsilon$ that do not satisfy \cref{assumption:epsilon}.

\begin{proposition}
    \label{res:optimization}
    Under \cref{assumption:epsilon}, we may satisfy all of the
    constraints in such a way that
    \[
        m(n,\epsilon)
        = O\left( \frac{n}{\epsilon \log n} \log^2\left(\tfrac{1}{\epsilon}\right) \right) \,.
    \]
\end{proposition}

    We first sketch some rough calculations to motivate the asymptotic dependence that each of
    $\ell$, $r$, $d$ and $m$ should have on $n$ and $\epsilon$. Combining
    \cref{constraint:right_tail,constraint:d_log}, we obtain
    \begin{equation}
        \label{eq:rough-m}
        m \gtrsim \frac{d}{r} \gtrsim \sqrt{\frac{1}{\ell r}} \log\left(\frac{1}{\epsilon}\right)
        \,.
    \end{equation}
    Thus, to allow $m$ to be as small as possible, we should make $\ell$ and $r$ as large as
    possible. We already have the bound $\ell \lesssim \frac{\epsilon}{n}$ from
    \cref{constraint:ell}. As we make $r$ larger, $d$ needs to grow due to \cref{constraint:d_log},
    which can only happen as long as the bound
    $d^6 9^d \left(\frac{r+\ell}{r-\ell}\right)^{2d-2} \leq \tfrac{1}{4} m(r-\ell)^2 n^2$ is
    preserved. The left-hand side of this inequality is exponential in $d$ while
    the right-hand side is polynomial in $n$, suggesting that we need $d = O(\log n)$. Since $d
    \gtrsim \sqrt{\frac{r}{\ell}} \log\left(\frac{1}{\epsilon}\right)$, this suggests setting
    \[
        r \approx \ell \cdot \frac{\log^2 n}{\log^2(1/\epsilon)} \,.
    \]
    Plugging back into \eqref{eq:rough-m}, we obtain
    \begin{equation}
    \label{eq:rough_m}
        m \gtrsim \frac{1}{\ell} \cdot \frac{\log^2(1/\epsilon)}{\log n}
        \approx \frac{n}{\epsilon \log n} \log^2\left(\frac{1}{\epsilon}\right) \,.
    \end{equation}
    \begin{remark}
      In \cref{section:improved-bound} we replace \cref{constraint:ell} with the looser
      \cref{constraint:better_ell} that $\ell \leq C_\ell \tfrac{\epsilon}{n} \log(1/\epsilon)$, which we
      can see from \cref{eq:rough_m} should improve $m$ by a $\log(1/\epsilon)$ factor.
    \end{remark}
    Let us now make this plan rigorous. (We have not tried too hard to optimize the constants.)

\begin{proof}
    We require \cref{assumption:epsilon} with $a \define 1/128$. We set
    \begin{itemize}
        \item $\ell \define C_\ell \frac{\epsilon}{n}$ for $C_\ell \define 1/20$.
        \item $r \define C_r \frac{\epsilon}{n} \frac{\log^2 n}{\log^2(1/\epsilon)}$, for constant
            $C_r \define 4 a^2 C_\ell$.
        \item $d \define \ceil{C_d \sqrt{\frac{r-\ell}{2\ell}} \log \left(\frac{20}{\epsilon}
            \right)}$, recalling that $C_d = 4\ln(2)$.
        \item $m \define 5.5 \frac{d}{r-\ell}$.
    \end{itemize}
    We claim that \cref{constraint:d_log,constraint:right_tail,constraint:variance,constraint:ell}
    are satisfied. \cref{constraint:ell,constraint:right_tail} are immediately satisfied by the
    choices of $m$ and $\ell$, while \cref{constraint:d_log} is satisfied by the choice of $d$, and
    because $r \ge 4\ell$ since
    \[
        \frac{r}{\ell} = \frac{C_r \log^2 n}{C_\ell \log^2(1/\epsilon)}
        > 4 a^2 \cdot \frac{\log^2 n}{(a \log n)^2}
        = 4 \,.
    \]
    Let us now verify \cref{constraint:variance}. Start with a calculation of
$\tfrac{d}{r-\ell}$, which also gives the claimed bound of $m = O(\tfrac{n}{\epsilon \log
n}\log^2(1/\epsilon))$:
    \begin{align*}
        \frac{d}{r-\ell}
        &\le \frac{2C_d \sqrt{\frac{r-\ell}{2\ell}} \log \left(\frac{20}{\epsilon} \right)}{r-\ell}
        = \frac{2C_d \log \left(\frac{20}{\epsilon} \right)}{\sqrt{2\ell(r-\ell)}}
        \le \frac{2C_d \log \left(\frac{20}{\epsilon} \right)}{\sqrt{2\ell \cdot r/2}}
        = \frac{2C_d \log \left(\frac{20}{\epsilon} \right)}{
            \sqrt{\frac{\epsilon}{n} \cdot C_\ell \cdot
                C_r \frac{\epsilon}{n} \frac{\log^2 n}{\log^2(1/\epsilon)}}} \\
        &= \frac{2C_d}{\sqrt{C_\ell C_r}} \cdot \frac{n}{\epsilon} \cdot \frac{\log(1/\epsilon) \log(20/\epsilon)}{\log n} 
        \leq \frac{12 C_d}{\sqrt{C_\ell C_r}} \cdot \frac{n}{\epsilon} \cdot \frac{\log^2(1/\epsilon)}{\log n} 
        = \frac{6 C_d}{C_\ell a} \cdot \frac{n}{\epsilon} \cdot \frac{\log^2(1/\epsilon)}{\log n} \,,
    \end{align*}
    the last inequality because $\epsilon < 1/2$ by \cref{assumption:epsilon}, so $\log(1/\epsilon)
    > 1$ and hence $\log(20/\epsilon) = \log(20) + \log(1/\epsilon) < 5 + \log(1/\epsilon) <
    6\log(1/\epsilon)$.

    To verify \cref{constraint:variance}, we start with checking the inequality $d^6 9^d
\left(\frac{r+\ell}{r-\ell}\right)^{2d-2} \leq \tfrac{1}{4} m(r-\ell)^2 n^2$. Using $m = 5.5
\frac{d}{r-\ell}$, the right-hand side is at least
    \begin{align*}
      \frac{1}{4} m (r-\ell)^2 n^2
      &\geq \frac{1}{4} \cdot 5.5 \cdot d (r - \ell) n^2 \,,
    \end{align*}
    so it suffices to prove
    \[
      \frac{d}{r-\ell} d^4 9^d \left(\frac{r+\ell}{r-\ell}\right)^{2d-2} \leq \frac{5.5}{4} n^2 \,.
    \]
    Using the upper bound on $\frac{d}{r-\ell}$ and the upper bound $\tfrac{r+\ell}{r-\ell} \leq
    \tfrac{(5/4)r}{(3/4)r} = 5/3$, it suffices to prove
    \begin{equation}
      \label{eq:sufficient-d-bound}
      d^4 5^{2(d-1)} \frac{6 C_d}{C_\ell} \frac{1}{a} \cdot \frac{n}{\epsilon}
        \frac{\log^2(1/\epsilon)}{\log n} \leq \frac{5.5}{4 \cdot 9} n^2 
      \iff
      5^{2(d-1)} \leq \frac{5.5 C_\ell}{4 \cdot 9 \cdot 6 C_d} \cdot a \cdot
        \frac{\epsilon n}{d^4} \cdot \frac{\log n}{\log^2(1/\epsilon)} \,.
    \end{equation}
    Upper bound $d$ by
    \begin{align*}
        d-1
        &\le C_d \sqrt{\frac{r-\ell}{2\ell}} \log \left(\frac{20}{\epsilon}\right)
        \le \frac{C_d}{\sqrt{2}} \sqrt{\frac{r}{\ell}} \cdot
            \log \left(\frac{20}{\epsilon}\right)
        \le \frac{C_d}{\sqrt{2}} \sqrt{\frac{C_r}{C_\ell}} \frac{\log(n)}{\log(1/\epsilon)}\cdot
            \log \left(\frac{20}{\epsilon}\right) \\
        &= \frac{2 C_d}{\sqrt{2}} \cdot a \cdot \frac{\log(n)}{\log(1/\epsilon)}\cdot
            \log \left(\frac{20}{\epsilon}\right)
        \leq 12\frac{C_d}{\sqrt{2}} \cdot a \cdot \log n \,.
    \end{align*}
    Lower bound the right-hand side of \eqref{eq:sufficient-d-bound} using the inequality above via
    $d \le 2(d-1)$, and $\epsilon \ge n^{-a}$, which implies that $\log(1/\epsilon) \le a \log n$,
    to obtain
    \begin{align*}
        \frac{5.5 C_\ell}{4 \cdot 9 \cdot 6 C_d} \cdot a \cdot
            \frac{\epsilon n}{d^4} \cdot \frac{\log n}{\log^2(1/\epsilon)}
        &\ge \frac{5.5 C_\ell}{4 \cdot 9 \cdot 6 C_d} \cdot a \cdot
            \frac{n^{1-a}}{\left(2 \cdot 12\frac{C_d}{\sqrt{2}} \cdot a \cdot \log n\right)^4}
            \cdot \frac{\log n}{a^2 \log^2 n} \\
        &= K \cdot \frac{1}{a^5} \cdot n^{1-a} \cdot \frac{1}{\log^5 n}
    \end{align*}
    for constant $K = \frac{5.5 \cdot 4 C_\ell}{4 \cdot 9 \cdot 6 \cdot 2^4 \cdot 12^4 C_d^5}$.
    Hence, returning to \eqref{eq:sufficient-d-bound}, it suffices to establish
    \[
      n^{\tfrac{24\log(5)}{\sqrt 2} C_d \cdot a} \leq K \cdot \frac{1}{a^5} \cdot n^{1-a} \cdot
        \frac{1}{\log^5 n} \,.
    \]
    We have $\frac{24 \log(5)}{\sqrt{2}} C_d a \le 0.8536$ while $\frac{K}{a^5} \ge 3.2$, and indeed
    the inequality
    \[
        n^{0.8536} \le 3.2 \cdot \frac{n^{127/128}}{\log^5 n}
    \]
    holds for sufficiently large $n$ (say, $n \ge 5.4 \cdot 10^{84}$).

    It remains to verify the inequality $m \leq \frac{1}{4^4} \epsilon^2 n^2$. We have
    \[
        m = 5.5 \frac{d}{r-\ell}
        \le 5.5 \cdot \frac{6 C_d}{C_\ell} \cdot \frac{1}{a} \cdot \frac{n}{\epsilon}
            \cdot \frac{\log^2(1/\epsilon)}{\log n} \,,
    \]
    so (since $\epsilon \geq n^{-a}$ and hence $\log(1/\epsilon) \le a \log n$) it suffices to have
    \begin{align*}
        &5.5 \cdot \frac{6 C_d}{C_\ell} \cdot \frac{1}{a} \cdot \frac{n}{\epsilon}
                \cdot \frac{\log^2(1/\epsilon)}{\log n}
            \le \frac{1}{4^4} \epsilon^2 n^2
        \impliedby
            \epsilon^3 \geq 4^4 \cdot 5.5 \cdot 6 \cdot \frac{C_d}{C_\ell} \cdot \frac{1}{a}
            \cdot \frac{1}{n} \cdot \frac{a^2 \log^2 n}{\log n}
        \\
        &\quad \impliedby
            \frac{n^{1-3a}}{\log n} \geq a \cdot 4^4 \cdot 5.5 \cdot 6 \cdot \frac{C_d}{C_\ell}
        \impliedby \frac{n^{125/128}}{\log n} \ge 3660 \,,
    \end{align*}
    which holds for sufficiently large $n$ (say, $n \ge 8 \cdot 10^4$).
\end{proof}

\subsection{Constraint IVb: Improvements from loosening constraint IV}
\label{section:improved-bound}

In this section we loosen \cref{constraint:ell} and obtain the improved bound of
$O\left(\tfrac{n}{\epsilon \log n}\log(1/\epsilon)\right)$. We replace the earlier constraint of
$\ell \leq C_\ell \tfrac{\epsilon}{n}$ with:

\begin{constraint}
\label{constraint:better_ell}
We require $\ell \leq C_\ell \tfrac{\epsilon}{n} \log(1/\epsilon)$ for any constant $C_\ell$
satisfying $C_\ell \leq \min\left\{\frac{C_d}{4\sqrt{3}}, \frac{1}{3}\right\}$.
\end{constraint}

To get the improved upper bound from this constraint, it is sufficient to replace the earlier
\cref{res:soundness} with:

\begin{lemma}
    \label{res:soundness_improved}
    Assume \cref{constraint:d_log,constraint:right_tail,constraint:better_ell}.  Then for all
    $\epsilon \leq 1/3$, and all $k \leq n$, if $p$ is $\epsilon$-far from being supported on $k$
    elements,
    \[
        \Ex{\bm{\widehat S}} > (1 + 3\epsilon/4) k \,.
    \]
\end{lemma}

Following the same arguments as in \cref{section:correctness_and_parameters} then leads to our main
result:
\begin{theorem}
\label{thm:formal-main}
For all $n \in \bN$ and $\epsilon \in (0,1)$, the sample complexity of testing support size is
at most
\[
  m(n,\epsilon) = O\left(\frac{n}{\epsilon \log n} \cdot \min\left\{ \log\tfrac{1}{\epsilon}, \log
n\right\} \right) \,.
\]
\end{theorem}
\begin{proof}
    The proof of correctness under
    \cref{constraint:d_log,constraint:right_tail,constraint:variance,constraint:better_ell} is the
    same as \cref{res:poisson-correctness}, replacing \cref{res:soundness} with
    \cref{res:soundness_improved}.

    The proof of the improved sample complexity follows from the same calculations as in
    \cref{res:optimization}: starting with the choice of parameters from \cref{res:optimization}, we
    \begin{enumerate}
        \item Scale $\ell$ and $r$ up by a factor of $\log(1/\epsilon)$.
        \item Scale $m$ down by a factor of $\log(1/\epsilon)$.
        \item Leave $d$ unchanged.
    \end{enumerate}
    It is then straightforward to see that each of
    \cref{constraint:d_log,constraint:right_tail,constraint:variance} is still satisfied under the
    new choices of parameters, and that \cref{constraint:better_ell} is indeed also satisfied.
\end{proof}

To show that \cref{constraint:better_ell} is sufficient, we first define a function $\Phi(\lambda)$
to quantify the worst-case behaviour of the sum $\sum_i Q(p_i)$
(\cref{section:better-bound-worst-case}), and then prove a lower bound on $\Phi(\lambda)$
(\cref{section:better-bound-lb}).

\paragraph*{Some notation.}
For $x \in (0,\ell)$, we will frequently write $x = \lambda \ell$
for $\lambda \in (0,1)$ and
\[
  \gamma = \gamma(\lambda) = \frac{2\alpha}{1-\alpha}(1-\lambda)
\]
which satisfies
\begin{equation}
\label{eq:psi-gamma}
  \psi(x) = \frac{r+\ell-2x}{r-\ell} = 1 + \frac{2\ell-2\lambda\ell}{r-\ell} = 1 +
\frac{2\ell(1-\lambda)}{r-\ell} = 1 + \frac{2\alpha}{1-\alpha}(1-\lambda) = 1 + \gamma \,.
\end{equation}

To simplify the analysis, we will define and analyze an auxiliary function $Q^*$ that lower bounds
$Q$.
Recall $Q(x) = 1 + e^{-mx} P_d(x)$.  We define the function
\[
  Q^*(x) \define \begin{cases}
    1 + P_d(x) &\text{ if } x < \ell \\
    1 - \delta &\text{ if } x \geq \ell \,,
  \end{cases}
\]
which satisfies:
\begin{proposition}[Properties of $Q^*$]
\label{res:Q*-properties}
$Q^*$ is continuous, concave and non-decreasing on $[0,1]$, and for all $x \in (0,\ell]$, $Q(x) \geq
Q^*(x)$.  Furthermore, if \cref{constraint:d_log,constraint:right_tail} hold, then $Q(x) \geq
Q^*(x)$ for all $x \in (0,1]$.
\end{proposition}
\begin{proof}
For $x \leq \ell$, $P_d(x)$ is concave and increasing
(\cref{res:pd-concave}), so $Q^*(x) = 1 + P_d(x)$ is also concave increasing. To ensure continuity,
it suffices to check the value of $Q^*$ at $x=\ell$. Observe that
\[
P_d(\ell) = - \delta T_d(\psi(\ell)) = - \delta T_d(1) = - \delta \,,
\]
so $Q^*(\ell) = 1-\delta =
1 + P_d(\ell)$, which makes $Q^*$ continuous. Then since $Q^*(x)$ is concave increasing on $x \leq
\ell$, and constant $Q^*(x) = Q^*(\ell)$ on $x \geq \ell$, it is concave and non-decreasing on $x
\in (0,1)$.

Now we show $Q(x) \geq Q^*(x)$. If $x \in (0, \ell)$, write $x = \lambda \ell$ so that $\psi(x) = 1
+ \gamma$. Then
\[
  P_d(x) = - \delta T_d(\psi(x)) = -\delta T_d(1 + \gamma) \,,
\]
so $P_d(x) \leq 0$ since $T_d(1+\gamma) \geq 0$ for $\gamma \geq 0$. Since $mx \geq 0$,
\[
  Q(x) = 1 + e^{-mx} P_d(x) \geq 1 + P_d(x) = Q^*(x) \,.
\]
If \cref{constraint:d_log,constraint:right_tail} hold, then for $x \geq \ell$, by
\cref{res:d_log,res:right_tail},
\[
  Q(x) \geq 1 - \delta = Q^*(x) \,. \qedhere
\]
\end{proof}

Since the quantity $\sum_i Q^*(p_i)$ is invariant under permutations of $p$, we may assume:
\begin{assumption}
\label{assumption:sorted}
Without loss of generality, we assume $p$ is sorted, so that $p_1 \geq p_2 \geq \dotsm$.
\end{assumption}

\subsubsection{Worst-case behaviour and the function \texorpdfstring{$\Phi(\lambda)$}{Phi(lambda)}}
\label{section:better-bound-worst-case}

We first show that the worst-case behaviour of $Q^*$ is captured by two variables: the
$n^\text{th}$-largest density $p_n$, and the total mass of elements lighter than $p_n$. Informally,
the worst-case distributions for our test statistic are those concentrating essentially all of their
mass at $p_n$.

\begin{proposition}
    \label{res:Q*-worst-case}
    Let $n \in \bN$ and let $p$ be a (sorted) distribution with $p_n > 0$. Let $\mu \define \sum_{i
    > n} p_i$. Then
    \[
        \sum_i Q^*(p_i) \ge \left(n + \frac{\mu}{p_n}\right) Q^*(p_n) \,.
    \]
\end{proposition}
\begin{proof}
    Recall that $Q^*$ is concave and non-decreasing by \cref{res:Q*-properties}. Since $Q^*$ is
    non-decreasing and $p_i \ge p_n$ for any $i \le n$, we have $Q^*(p_i) \ge Q^*(p_n)$ for any $i
    \le n$. On the other hand, for any $i > n$, we have $p_i \le p_n$ and then, using the concavity
    of $Q^*$,
    \[
        Q^*(p_i)
        = Q^*\left( \frac{p_i}{p_n} \cdot p_n + \left(1-\frac{p_i}{p_n}\right) \cdot 0 \right)
        \ge \frac{p_i}{p_n} \cdot Q^*(p_n) + \left(1-\frac{p_i}{p_n}\right) \cdot Q^*(0)
        = \frac{p_i}{p_n} \cdot Q^*(p_n) \,,
    \]
    the last equality since $Q^*(0) = 1 + P_d(0) = 0$. We conclude that
    \[
        \sum_i Q^*(p_i)
        \ge \sum_{i \le n} Q^*(p_n) + \sum_{i > n} \frac{p_i}{p_n} \cdot Q^*(p_n)
        = \left(n + \frac{\mu}{p_n}\right) Q^*(p_n) \,. \qedhere
    \]
\end{proof}

We can now transform the problem of lower bounding $\sum_i Q^*(p_i)$ into the problem of lower
bounding another function $\Phi(\lambda)$ defined as:

\begin{definition}
For $\lambda \in (0,1]$, define
\[
  \Phi(\lambda) \define \left(1 + \frac{\epsilon}{\lambda \ell n}\right) Q^*(\lambda \ell) \,.
\]
\end{definition}

We will show a lower bound of $\Phi(\lambda) \geq 1+3\epsilon/4$ below, in
\cref{section:better-bound-lb}. Here we show why it works:

\begin{lemma}
    \label{res:Q*_bound_from_phi}
    Assume \cref{constraint:d_log}. Suppose $\Phi(\lambda) \geq 1 + 3\epsilon/4$ for all $\lambda
    \in (0,1)$. Then for any $k \leq n$ and any $p$ that is $\epsilon$-far from having $|\supp(p)| \leq k$,
    \[
        \sum_i Q^*(p_i) > (1+3\epsilon/4) k \,.
    \]
\end{lemma}
\begin{proof}
    Since $p$ is sorted and $\epsilon$-far from having $|\supp(p)| \le k$, we have $p_k > 0$ and
    $\sum_{i > k} p_i > \epsilon$. We also have $p_k \le 1/k$ since $p$ is a probability
    distribution. We consider two cases: $p_k \ge \ell$ and $p_k < \ell$.

    First, suppose $p_k \ge \ell$. Then $Q^*(p_k) = 1-\delta$ by definition, so by
    \cref{res:Q*-worst-case},
    \[
        \sum_i Q^*(p_i)
        > \left( k + \frac{\epsilon}{1/k} \right) (1-\delta)
        = (1+\epsilon) (1-\delta) k
        \ge (1+3\epsilon/4) k \,,
    \]
    the last inequality since $\delta \le \epsilon/20 < \epsilon/8$ by \cref{res:d_log}. Otherwise,
    if $p_k = \lambda \ell$ for some $\lambda \in (0, 1)$, then \cref{res:Q*-worst-case} gives
    \[
        \sum_i Q^*(p_i)
        > \left( k + \frac{\epsilon}{\lambda \ell} \right) Q^*(\lambda \ell)
        = k \left(1 + \frac{\epsilon}{\lambda k \ell} \right) Q^*(\lambda \ell)
        \geq \Phi(\lambda) \cdot k
        \ge (1+3\epsilon/4) k \,. \qedhere
    \]
\end{proof}

\subsubsection{Lower bound on \texorpdfstring{$\Phi(\lambda)$}{Phi(lambda)}}
\label{section:better-bound-lb}

Our goal is to prove $\Phi(\lambda) \geq 1+3\epsilon/4$. \cref{fig:phi} shows a plot of
$\Phi(\lambda)$ compared to $1 + \epsilon$ for certain example settings of parameters.

\begin{figure}[h!]
    \centering
    \begin{subfigure}{0.46\textwidth}
        \centering
        \includegraphics[width=\textwidth]{"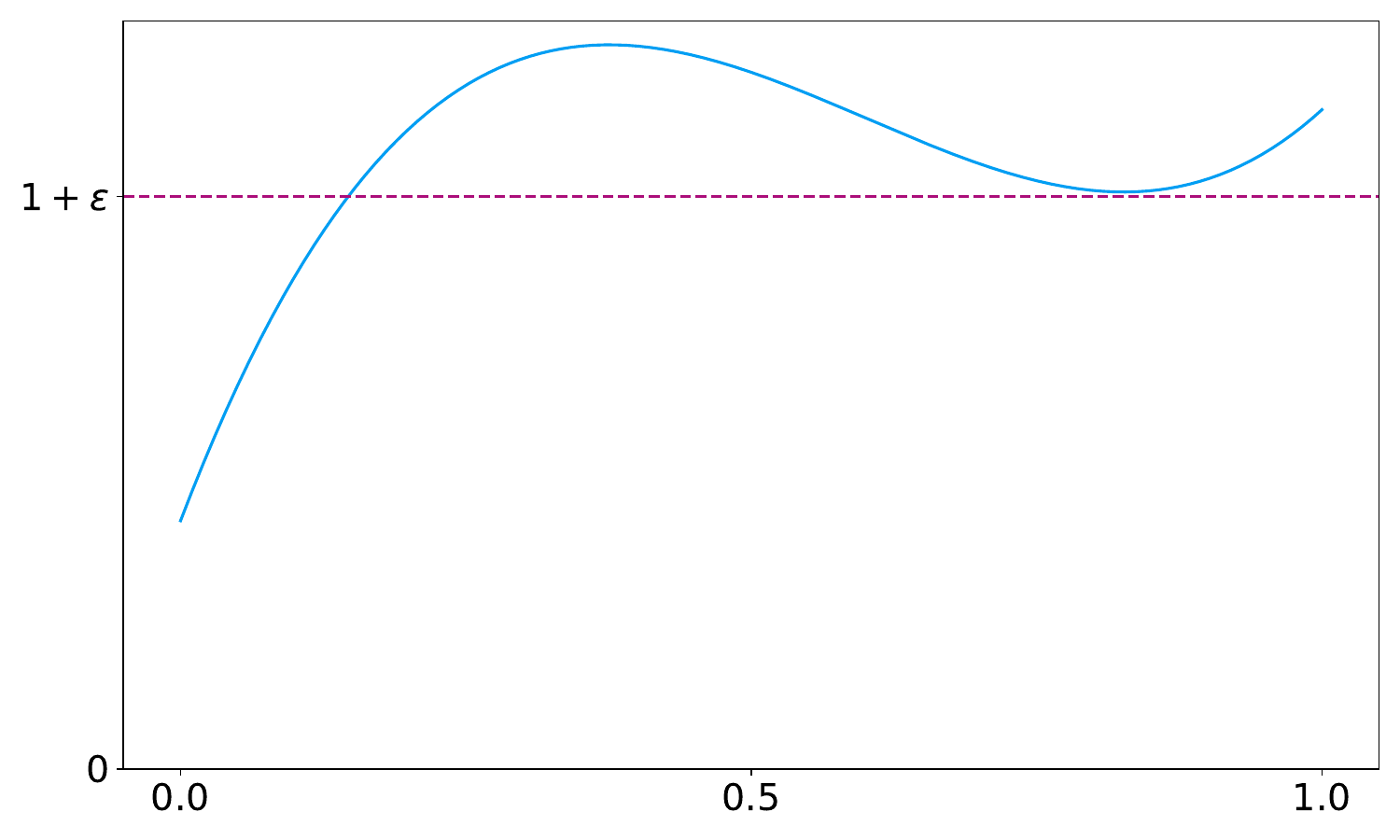"}
        \caption{$\Phi(\lambda)$ when $d$ is too small.}
        \label{fig:phi_bad}
    \end{subfigure}
    \hspace{0.05\textwidth} 
    \begin{subfigure}{0.46\textwidth}
        \centering
        \includegraphics[width=\textwidth]{"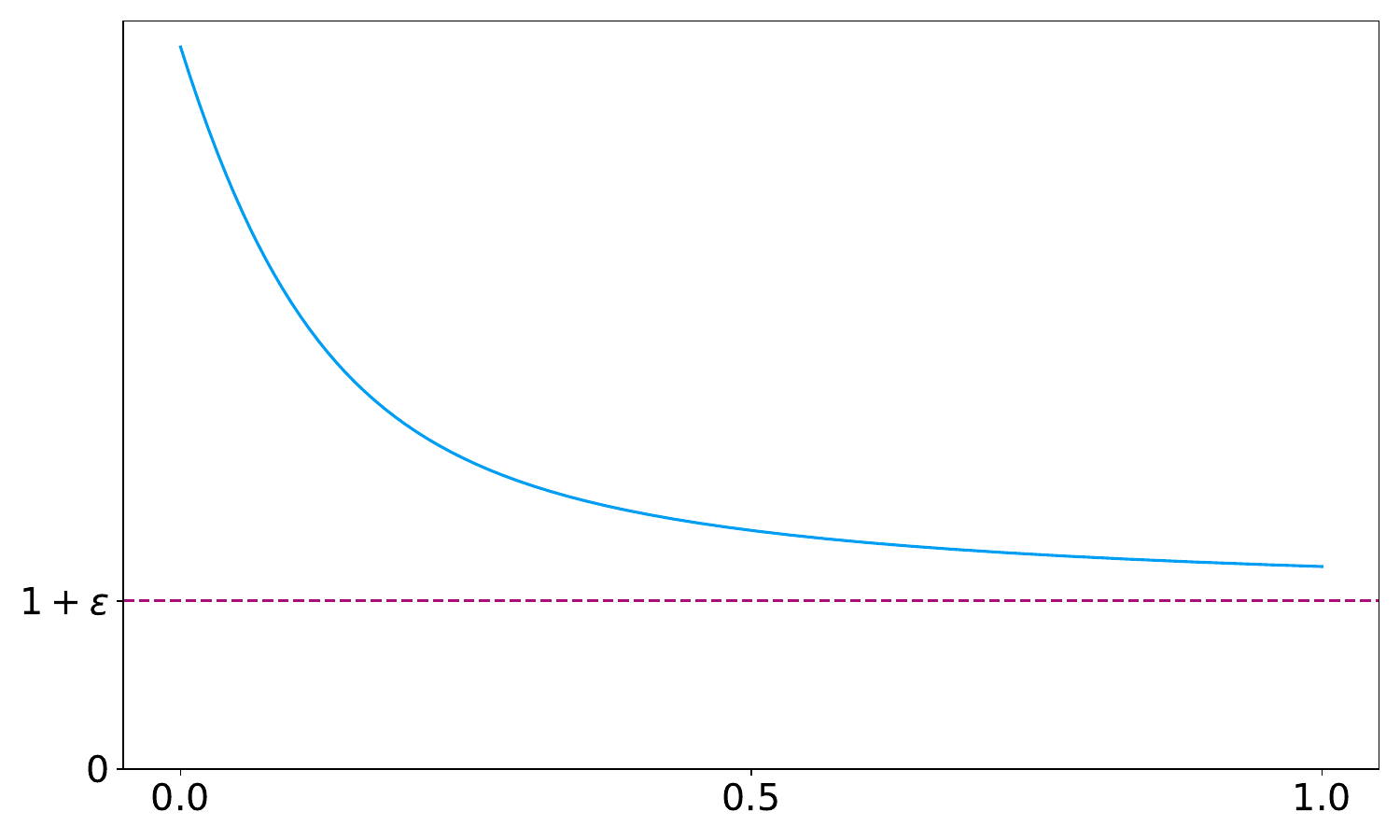"}
        \caption{$\Phi(\lambda)$ under \cref{constraint:d_log,constraint:better_ell}.}
        \label{fig:phi_good}
    \end{subfigure}
    \caption{The function $\Phi(\lambda)$ with bad and good parameters.}
    \label{fig:phi}
\end{figure}

Note that $\Phi(\lambda) = \left(1 + \frac{\epsilon}{\lambda \ell n}\right) Q^*(\lambda \ell)$ is
undefined on $\lambda = 0$. We will proceed in two steps:
\begin{enumerate}
\item Show that $\Phi(\lambda) \ge 1+3\epsilon/4$ at $\lambda = 1$ and in the limit $\lambda \to
    0^+$;
\item Show that $\Phi(\lambda) \ge 1+3\epsilon/4$ within $\lambda \in (0,1)$ by studying the
    derivative $\Phi'(\lambda)$.
\end{enumerate}

The case $\lambda = 1$ is the easiest to handle:

\begin{proposition}
    \label{res:phi-at-1}
    Assume \cref{constraint:d_log,constraint:better_ell}. Then for all $n \in \bN$ and $\epsilon \in
    (0, 1)$, $\Phi(1) \ge 1 + 3\epsilon/4$.
\end{proposition}
\begin{proof}
    Note that, by \cref{res:d_log}, we have $\delta \le \epsilon/20 < \epsilon/8$. When $\lambda =
    1$, we have $Q^*(\lambda \ell) = 1-\delta$ by definition and hence, using the inequality $\ell
    \le \tfrac{\epsilon}{n} \log(1/\epsilon)$ obtained from \cref{constraint:better_ell},
    \[
        \Phi(1)
        = \left(1 + \frac{\epsilon}{\ell n}\right) Q^*(\ell)
        \ge \left(1 + \frac{1}{\log(1/\epsilon)}\right) (1-\delta)
        > (1 + \epsilon) (1-\delta)
        \ge 1 + 3\epsilon/4 \,. \qedhere
    \]
\end{proof}

To reason about $\lim_{\lambda \to 0^+} \Phi(\lambda)$, we will use L'Hôpital's rule, for which we
will need bounds on the derivative $T_d'(x)$ of the Chebyshev polynomials. We state the following
bound, whose proof proceeds by a direct calculation and is deferred to \cref{section:td-prime}.

\begin{restatable}{proposition}{restdprimelb}
    \label{res:td-prime-lb}\RestateRemark
    For any $d \in \bN$ and $\gamma \in (0,1)$, we have
    \[
        T_d'(1+\gamma) \ge \frac{d}{\sqrt{3\gamma}} \left(T_d(1+\gamma) - 1\right) \,.
    \]
\end{restatable}

We may now complete Step 1.

\begin{proposition}
    \label{res:phi-limit-0}
    Assume \cref{constraint:d_log,constraint:better_ell}. Then the limit $\lim_{\lambda \to 0^+}
    \Phi(\lambda)$ exists, is finite, and satisfies
    \begin{equation}
        \label{eq:phi-lb}
        \lim_{\lambda \to 0^+} \Phi(\lambda) \ge 2 \,.
    \end{equation}
\end{proposition}
\begin{proof}
    Recall from \cref{eq:psi-gamma} that $Q^*(\lambda \ell) = 1 + P_d(\lambda \ell) = 1 -
    \delta T_d(1 + \gamma)$. Then 
    \begin{align*}
        \lim_{\lambda \to 0^+} \Phi(\lambda)
        &= \lim_{\lambda \to 0^+}
            \left(1 + \frac{\epsilon}{\lambda \ell n}\right)
            \left(1 - \delta T_d(1 + \gamma)\right) \\
        &= \lim_{\lambda \to 0^+} \left[ 1 - \delta T_d(1 + \gamma) \right]
            + \frac{\epsilon}{\ell n} \cdot \lim_{\lambda \to 0^+} \left[
                \frac{1 - \delta T_d(1 + \gamma)}{\lambda} \right] \,.
    \end{align*}
    Since $\gamma = \frac{2\alpha}{1-\alpha}(1-\lambda)$, and by the definition of $\delta$, the
    first limit evaluates to
    \[
        1 - \delta T_d\left(1 + \frac{2\alpha}{1-\alpha}\right) = 1 - \delta \cdot \frac{1}{\delta}
        = 0 \,.
    \]
    As for the second limit, note that both the numerator (for the same reason) and denominator go
    to $0$ as $\lambda \to 0^+$. Applying L'Hôpital's rule,
    \begin{align*}
        \lim_{\lambda \to 0^+} \Phi(\lambda)
        &= \frac{\epsilon}{\ell n} \cdot \lim_{\lambda \to 0^+}
        \frac{\frac{\odif}{\odif \lambda} (1 - \delta T_d(1 + \gamma))}
             {\frac{\odif}{\odif \lambda} \lambda}
        = -\frac{\delta \epsilon}{\ell n} \cdot \lim_{\lambda \to 0^+}
            \frac{\odif}{\odif \lambda} T_d\left(1 + \frac{2\alpha}{1-\alpha}(1-\lambda)\right) \\
        &= \frac{\delta \epsilon}{\ell n} \cdot \frac{2\alpha}{1-\alpha} \cdot
            \lim_{\lambda \to 0^+} T_d'\left(1 + \frac{2\alpha}{1-\alpha}(1-\lambda)\right)
        = \frac{\delta \epsilon}{\ell n} \cdot \frac{2\alpha}{1-\alpha} \cdot
            T_d'\left(1 + \frac{2\alpha}{1-\alpha}\right) \,,
    \end{align*}
    where the last equality holds since $T_d'(x)$ is a polynomial and hence continuous, and this
    also establishes the existence and finiteness of the limit. Using \cref{res:td-prime-lb} and the
    definition of~$\delta$,
    \begin{align*}
        \lim_{\lambda \to 0^+} \Phi(\lambda)
        &\ge \frac{\delta \epsilon}{\ell n} \cdot \frac{2\alpha}{1-\alpha}
            \cdot \frac{d}{\sqrt{3 \cdot \frac{2\alpha}{1-\alpha}}}
            \cdot \left(\frac{1}{\delta} - 1\right)
        = \sqrt{\frac{2\alpha}{1-\alpha}} \cdot
            \frac{\delta \epsilon d}{\ell n \sqrt{3}} \cdot \left(\frac{1}{\delta} - 1\right) \,.
    \end{align*}
    Since $\delta \le 1/2$ by \cref{res:d_log}, we have $\frac{1}{\delta} - 1 \ge
    \frac{1}{2\delta}$. Using \cref{constraint:d_log,constraint:better_ell} and the definition
    $\alpha = \ell/r$, we obtain
    \[
        \lim_{\lambda \to 0^+} \Phi(\lambda)
        \ge \sqrt{\frac{2\alpha}{1-\alpha}} \cdot
            \frac{\delta \epsilon}{n\sqrt{3}}
            \cdot \frac{n}{C_\ell \epsilon \log(1/\epsilon)}
            \cdot C_d \sqrt{\frac{1-\alpha}{2\alpha}} \log(1/\epsilon)
            \cdot \frac{1}{2\delta}
        \ge 2 \,. \qedhere
    \]
\end{proof}

\begin{remark}
\label{remark:no-better-parameter}
The last step in the above proof shows that \cref{constraint:better_ell} is the best possible
relaxation of our constraint on $\ell$ for the present proof strategy: if we had $\ell \gg
\tfrac{\epsilon}{n} \log(1/\epsilon)$, then we would only have obtained a $o(1)$ lower bound on
$\lim_{\lambda \to 0^+} \Phi(\lambda)$, but we require a $1 + 3\epsilon/4$ lower bound.
\end{remark}

Thanks to \cref{res:phi-limit-0}, we hereafter consider the continuous extension $\Phi : [0,1] \to
\bR$, \ie we define $\Phi(0) \define \lim_{\lambda \to 0^+} \Phi(\lambda)$.

The next step is to show that $\Phi$ satisfies a certain differential inequality, which will help us
conclude that $\Phi(\lambda) \ge 1 + 3\epsilon/4$ at any critical points where $\Phi'(\lambda) = 0$.
It is convenient to define $L \define \tfrac{\ell n}{\epsilon}$, so that
\[
  \Phi(\lambda) = \left(1 + \frac{1}{L \lambda}\right) (1 - \delta T_d(1+\gamma)) \,.
\]

\begin{lemma}
    \label{res:phi-differential-ineq}
    Assume \cref{constraint:d_log}. Define $A \define \sqrt{\tfrac{1}{3}} \cdot C_d
\log\left(\tfrac{1}{\epsilon}\right)$. Then for all $\lambda \in (0, 1)$,
    \[
        \Phi'(\lambda)
        \ge -\Phi(\lambda)
            \left( A + \frac{1}{\lambda (L\lambda + 1)} \right)
            + (1-\delta)A \left( 1 + \frac{1}{L\lambda} \right) \,.
    \]
\end{lemma}
\begin{proof}
    Recall that $\gamma = \frac{2\alpha}{1-\alpha}(1-\lambda)$ so $\tfrac{\odif}{\odif
\lambda}\gamma = -\frac{2\alpha}{1-\alpha}$. Then
    \begin{align*}
        \Phi'(\lambda) 
        &= \frac{\odif}{\odif \lambda} \left[
            \left(1 + \frac{1}{L \lambda}\right)
            \left(1 - \delta T_d\left(1 + \gamma\right)\right)
        \right] \\
        &= -\frac{1}{L \lambda^2}
                \left(1 - \delta T_d\left(1 + \gamma \right)\right)
            + \left(1 + \frac{1}{L \lambda}\right) \cdot (-\delta) \cdot
                T_d'\left(1 + \gamma \right) \cdot
                \left( -\frac{2\alpha}{1-\alpha} \right) \\
        &= \frac{-\frac{1}{L \lambda^2}}{1 + \frac{1}{L\lambda}} \Phi(\lambda)
            + \delta \cdot \frac{2\alpha}{1-\alpha} \cdot
            \left(1 + \frac{1}{L \lambda}\right) \cdot T_d'(1+\gamma) \\
        &= -\frac{1}{\lambda (L\lambda+1)} \Phi(\lambda)
            + \delta \cdot \frac{2\alpha}{1-\alpha} \cdot
            \left(1 + \frac{1}{L \lambda}\right) \cdot T_d'(1+\gamma) \,.
    \end{align*}
    Using \cref{res:td-prime-lb,constraint:d_log} and recalling that $T_d(1+\gamma) \ge 1$,
    \begin{align*}
        \Phi'(\lambda)
        &\ge -\frac{1}{\lambda (L\lambda+1)} \Phi(\lambda)
            + \delta \cdot \frac{2\alpha}{1-\alpha} \cdot
            \left(1 + \frac{1}{L \lambda}\right) \cdot
                \frac{d}{\sqrt{3 \cdot \frac{2\alpha}{1-\alpha}(1-\lambda)}}
                \left(T_d(1+\gamma) - 1\right) \\
        &\ge -\frac{1}{\lambda (L\lambda+1)} \Phi(\lambda)
            + \delta \cdot \sqrt{\frac{2\alpha}{1-\alpha}} \cdot
            \left(1 + \frac{1}{L \lambda}\right) \cdot
                \frac{C_d \sqrt{\frac{1-\alpha}{2\alpha}} \log \left(\frac{1}{\epsilon} \right)}{\sqrt{3}}
                \left(T_d(1+\gamma) - 1\right) \\
        &= -\frac{1}{\lambda (L\lambda+1)} \Phi(\lambda)
            + \delta \cdot
            \left(1 + \frac{1}{L \lambda}\right) \cdot
                \frac{C_d \log \left(\frac{1}{\epsilon} \right)}{\sqrt{3}}
                \left(T_d(1+\gamma) - 1\right) \\
        &= -\frac{1}{\lambda (L\lambda+1)} \Phi(\lambda)
            + A \left(1 + \frac{1}{L \lambda}\right) \cdot \delta T_d(1+\gamma)
            - \delta A \left(1 + \frac{1}{L \lambda}\right) \,.
    \end{align*}
    Now, we may use the definition of $\Phi$ to rewrite the middle term:
    \[
        \Phi(\lambda) =
        \left(1 + \frac{1}{L \lambda}\right) \left(1 - \delta T_d(1 + \gamma)\right)
        \implies
        \delta T_d(1 + \gamma)
        = 1 - \frac{\Phi(\lambda)}{1 + \frac{1}{L \lambda}} \,.
    \]
    Therefore
    \begin{align*}
        \Phi'(\lambda)
        &\ge -\frac{1}{\lambda (L\lambda+1)} \Phi(\lambda)
            + A \left(1 + \frac{1}{L \lambda}\right)
                \left( 1 - \frac{\Phi(\lambda)}{1 + \frac{1}{L \lambda}} \right)
            - \delta A \left(1 + \frac{1}{L \lambda}\right) \\
        &= -\frac{1}{\lambda (L\lambda+1)} \Phi(\lambda)
            + A \left( 1 + \frac{1}{L\lambda} - \Phi(\lambda) \right)
            - \delta A \left(1 + \frac{1}{L \lambda}\right) \\
        &= -\Phi(\lambda) \left( A + \frac{1}{\lambda (L\lambda+1)} \right)
            + (1-\delta) A \left( 1 + \frac{1}{L\lambda} \right) \,. \qedhere
    \end{align*}
\end{proof}

We may now conclude the argument.

\begin{lemma}
    \label{res:phi_bound}
Assume \cref{constraint:d_log,constraint:better_ell}. Then for all $n \in \bN$ and $\epsilon \leq
1/3$,
\[
  \forall \lambda \in [0,1] \;:\; \Phi(\lambda) \ge 1 + 3\epsilon/4 \,.
\]
\end{lemma}
\begin{proof}
    Note that \cref{res:phi-at-1,res:phi-limit-0} already give the result for $\lambda = 0$ and
    $\lambda = 1$, and moreover $\Phi$ is continuous on $[0,1]$ (by
    \cref{res:Q*-properties}) and differentiable in $(0,1)$ (by the differentiability of $Q^*$ in
    $(0, \ell)$).
    Thus it remains
    to check points $\lambda \in (0,1)$ for which $\Phi'(\lambda) = 0$, if any exist. Suppose
    $\lambda \in (0,1)$ satisfies $\Phi'(\lambda) = 0$. Then \cref{res:phi-differential-ineq}
    implies that
    \[
        0 = \Phi'(\lambda)
        \ge -\Phi(\lambda)
            \left( A + \frac{1}{\lambda (L\lambda + 1)} \right)
            + (1-\delta)A \left( 1 + \frac{1}{L\lambda} \right) \,,
    \]
    where again $A \define \sqrt{\frac{1}{3}} \cdot C_d \log\left(\frac{1}{\epsilon}\right)$.
    Rearranging, we obtain
    \[
        \Phi(\lambda)
        \ge \frac{(1-\delta)A \left( 1 + \frac{1}{L\lambda} \right)}
                 {\left( A + \frac{1}{\lambda (L\lambda + 1)} \right)}
        = \frac{(1-\delta)A \left( \frac{L\lambda + 1}{L\lambda} \right)}
                 {\left(\frac{A \lambda (L\lambda + 1) + 1}{\lambda (L\lambda + 1)} \right)}
        = \frac{(1-\delta) (AL^2\lambda^2 + 2AL\lambda + A)}
                 {AL^2\lambda^2 + AL\lambda + L} \,.
    \]
    Take $K \define A/L \geq \tfrac{C_d}{\sqrt{3}C_\ell}$ which satisfies $K \geq 4$ 
under \cref{constraint:better_ell}. Then the right-hand side is
\begin{align*}
  (1-\delta)\left(1 + \frac{KL^2\lambda + (K-1)L}{KL^3\lambda^2 + KL^2\lambda + L}\right)
  &= (1-\delta)\left(1 + \frac{KL\lambda + (K-1)}{KL^2\lambda^2 + KL\lambda + 1}\right) \,.
\end{align*}
If $L\lambda < 1/K$ then, by preserving only the term $(K-1)$ in the numerator and using the
bound $K L^2 \lambda^2 \le KL\lambda \le 1$ in the denominator, this is at least
\[
  (1-\delta)\left(1 + \frac{K-1}{3}\right) \geq (1-\delta) \cdot 2 \ge 1 + 3\epsilon/4 \,,
\]
the last inequality since $\delta \le \epsilon/20$ by \cref{res:d_log}. Otherwise, by preserving
only the term $KL\lambda$ in the numerator and upper bounding each term in the denominator by the
maximum of the three, it is at least
\begin{align*}
  (1-\delta)\left(1 + KL\lambda \cdot
    \min\left\{ \frac{1}{3KL^2\lambda^2}, \frac{1}{3KL\lambda}, \frac{1}{3} \right\} \right)
&\geq 
  (1-\delta)\left(1 + \min\left\{ \frac{1}{3L}, \frac{1}{3} \right\} \right) \\
&\ge (1-\delta)(1+\epsilon)
\ge 1 + 3\epsilon/4 \,,
\end{align*}
where the first inequality holds since $3KL\lambda \ge 3$ and $\lambda \le 1$, and the
second inequality holds since $L \leq C_\ell \log(1/\epsilon) \leq
\frac{1}{3\epsilon}$ by \cref{constraint:better_ell}.
\end{proof}

We now have all the ingredients to conclude the soundness of the improved tester, and hence finish
the proof of the main theorem.

\begin{proof}[Proof of Lemma~\ref{res:soundness_improved}]
    Combine \cref{eq:mean-Q,res:Q*-properties,res:Q*_bound_from_phi,res:phi_bound} to obtain
    \[
        \Ex{\bm{\widehat S}}
        = \sum_{i \in \bN} Q(p_i)
        \ge \sum_{i \in \bN} Q^*(p_i)
        > (1 + 3\epsilon/4) k \,. \qedhere
    \]
\end{proof}

\section{An Effective Lower Bound on Support Size}
\label{section:good-lower-bound}

To prove \cref{res:intro-good-lb} from the introduction, which we restate here for convenience,
we need to essentially do a binary search for the correct parameters, in order to control the
variance (since the variance depends on the parameter $n$).

\begin{corollary}
    There exists an algorithm $A$ which, given inputs $n \in \bN$, $\epsilon \in (0,1/3)$ and sample
    access to unknown distribution $p$ over $\bN$, draws $O\left(\frac{n}{\epsilon \log n} \cdot
    \min\left\{ \log(1/\epsilon), \log n \right\} \right)$ independent samples from $p$ and outputs
    a number $\widehat{\bm S}$ which satisfies (with probability at least $3/4$ over the samples)
\[
  \min\{ \eff_\epsilon(p), n\} \leq \widehat{\bm{S}} \leq (1+\epsilon)|\supp(p)| \,.
\]
\end{corollary}
\begin{proof}
The algorithm proceeds as follows:
\begin{enumerate}
\item For $i=0$ to $\log n$, perform the following.
\begin{enumerate}
\item Set $n_i \define n/2^i$ and $\delta_i = \tfrac{1}{4 \cdot 2^{i+1}}$.
\item If $n_i, \epsilon$ do not satisfy \cref{assumption:epsilon}, use $O(m(n_i, \epsilon) \cdot
    \log(1/\delta_i))$ samples to obtain an estimate $\widehat{\bm{S}_i}$ satisfying the conditions
    of \cref{res:naive-estimator} with probability at least $1-\delta_i$ (using the median trick to
    boost the error). Output $\widehat{\bm{S}_i}$ and terminate.
\item Otherwise, if $n_i, \epsilon$ satisfy \cref{assumption:epsilon}, use $O(m(n_i, \epsilon) \cdot
    \log(1/\delta_i))$ samples obtain an estimate of the test statistic $\widehat{\bm{S}_i}$ from
    \cref{def:test-statistic}, which (again using the median trick) satisfies the conditions of
    \cref{res:poisson-correctness} with probability at least $1-\delta_i$. If $\widehat{\bm{S}_i}
    \geq n_{i+1}$ output $\widehat{\bm{S}_i}$ and terminate. Otherwise continue.
\end{enumerate}
\item If the algorithm did not terminate in $\log n$ steps, then output $1$.
\end{enumerate}
The number of samples used by the algorithm is at most (for some constant $C > 0$):
\begin{align*}
&C \cdot \sum_{i=0}^{\log n} \frac{n_i}{\log(n_i)} \log\left(\frac{1}{\delta_i}\right)
\frac{1}{\epsilon} \cdot \min\{ \log(1/\epsilon), \log(n_i) \} \\
&\qquad= C \cdot \frac{1}{\epsilon}\sum_{i=0}^{\tfrac{1}{2} \log n} \frac{n/2^i}{\log(n)-i}
    (i + 3) \cdot \min\{ \log(1/\epsilon),
\log(n) - i \} \\
&\qquad\qquad+ C \cdot \frac{1}{\epsilon}\sum_{i=1+\tfrac{1}{2}\log n}^{\log n}
    \frac{n/2^i}{\log(n)-i} (i + 3) \cdot \min\{ \log(1/\epsilon),
\log(n) - i \} \\
&\qquad\leq C \cdot \frac{2n}{\epsilon \log n} \min\{\log(1/\epsilon), \log(n)\} \sum_{i=0}^{\tfrac{1}{2}
\log n} \frac{1}{2^i} (i + 3) + 2C \cdot \frac{1}{\epsilon}\sqrt n \log n \\
&\qquad= O\left( \frac{n}{\epsilon \log n} \min\{ \log(1/\epsilon), \log n \} \right).
\end{align*}

The probability that any of the estimates $\widehat{\bm{S}_i}$ fails to satisfy the conditions in
\cref{res:poisson-correctness} or \cref{res:naive-estimator} (whichever corresponds to the estimator
used in the $i^{\text{th}}$ step) is at most $\sum_{i=0}^{\log n} \delta_i \leq \frac{1}{4}$.
Assuming that every $\widehat{\bm{S}_i}$ satisfies its corresponding condition, we separately prove
the upper and lower bounds on the output of the algorithm.

\begin{claim}
    The output satisfies $\widehat{\bm{S}} \le (1+\epsilon)|\supp(p)|$.
\end{claim}
\begin{subproof}
    If the algorithm reaches its last step and outputs $1$, there is nothing to show. Similarly, if
    it outputs an estimate $\widehat{\bm{S}_i}$ coming from \cref{res:naive-estimator}, then
    $\widehat{\bm{S}_i} \le |\supp(p)|$ and we are done.

    The remaining case is that the algorithm terminates at some step $i^*$ by outputting an estimate
    $\widehat{\bm{S}_{i^*}}$ satisfying the conditions of \cref{res:poisson-correctness}. We claim
    that $n_{i^*} \le 3|\supp(p)|$. Indeed, assuming that $n_{i^*} > 3|\supp(p)|$ for a
    contradiction, \cref{res:poisson-correctness} implies that
    \begin{align*}
        \widehat{\bm{S}_i}
        &< (1 + \epsilon/4)|\supp(p)| + \frac{\epsilon n_{i^*}}{4}
        < (1 + \epsilon/4) \cdot \frac{n_{i^*}}{3} + \frac{\epsilon n_{i^*}}{4} \\
        &= n_{i^*} \left(\frac{1}{3} + \frac{\epsilon}{12} + \frac{\epsilon}{4}\right)
        = n_{i^*} \left(\frac{1}{3} + \frac{\epsilon}{3}\right)
        < \frac{n_{i^*}}{2}
        = n_{i^*+1} \,,
    \end{align*}
    contradicting the assumption that the algorithm terminates at step $i^*$. Hence $n_{i^*} \le
    3|\supp(p)|$. Thus \cref{res:poisson-correctness} gives
    \[
        \widehat{\bm{S}_i} < (1 + \epsilon/4)|\supp(p)| + \frac{\epsilon n_{i^*}}{4}
        \le (1 + \epsilon/4)|\supp(p)| + \frac{3\epsilon}{4} |\supp(p)|
        = (1 + \epsilon) |\supp(p)| \,.
        \qedhere
    \]
\end{subproof}

\begin{claim}
    The output satisfies $\widehat{\bm{S}} \ge \min\{\eff_\epsilon(p), n\}$.
\end{claim}
\begin{subproof}
    We first consider the edge case where $\eff_\epsilon(p) = 1$. In this case, the output satisfies
    $\widehat{\bm{S}} \ge \eff_\epsilon(p)$ since the algorithm always outputs at least
    $1$. Going forward, suppose $\eff_\epsilon(p) \ge 2$.

    Now let $i^*$ be the smallest non-negative integer such that $\eff_\epsilon(p)-1 \geq
    n_{i^*+1}$, which exists because $n_{\log(n)+1} = \frac{1}{2} \le \eff_\epsilon(p)-1$, so $i^* =
    \log n$ satisfies the condition. We consider two cases.

    First, suppose that $\eff_\epsilon(p)-1 \ge n_{i^*}$, which implies that $i^* = 0$ by the
    minimality of $i^*$. Then $n = n_0 \le \eff_\epsilon(p)-1$, so $p$ is $\epsilon$-far from having
    support size $n$ by \cref{obs:eff}, and we conclude that the algorithm terminates with output
    satisfying $\widehat{\bm{S}_0} \ge (1+\epsilon/2) n$ (by \cref{res:poisson-correctness}) or
    $\widehat{\bm{S}_0} \ge n$ (by \cref{res:naive-estimator}).

    Otherwise, we have that $n_{i^*} > \eff_\epsilon(p)-1 \geq n_{i^*+1}$. For all $j < i^*$, if the
    algorithm terminates on loop $j$ then it outputs $\widehat{\bm{S}_j} \geq n_{j+1} \geq n_{i^*} >
    \eff_\epsilon(p)-1$, so that $\widehat{\bm{S}_j} \ge \eff_\epsilon(p)$.

    Finally, suppose the algorithm reaches loop $i^*$. If $n_{i^*}, \epsilon$ do not satisfy
    \cref{assumption:epsilon}, then the output $\widehat{\bm{S}_{i^*}}$ satisfies
    $\widehat{\bm{S}_{i^*}} \ge \min\{\eff_\epsilon(p), n_{i^*}\} = \eff_\epsilon(p)$ by
    \cref{res:naive-estimator}.

    Otherwise, $n_{i^*}, \epsilon$ satisfy \cref{assumption:epsilon}, which implies that
    $\eff_\epsilon(p)-1 \ge n_{i^*+1} > 4/\epsilon$; then, since $\min\{\eff_\epsilon(p)-1,
    n_{i^*}\} = \eff_\epsilon(p)-1$, we conclude from \cref{res:poisson-correctness} that
    \begin{align*}
      \widehat{\bm{S}_{i^*}}
      &> (1+3\epsilon/4) (\eff_\epsilon(p)-1) - \epsilon n_{i^*}/4
      = (1+3\epsilon/4) (\eff_\epsilon(p)-1) - \epsilon n_{i^*+1} / 2 \\
      &\geq (1+\epsilon/4) (\eff_\epsilon(p)-1)
      > \eff_\epsilon(p) - 1 + \frac{\epsilon}{4} \cdot \frac{4}{\epsilon}
      = \eff_\epsilon(p) \,,
    \end{align*}
    which is larger than $n_{i^*+1}$, so the algorithm terminates with output satisfying
    $\widehat{\bm{S}_{i^*}} > \eff_\epsilon(p)$.
\end{subproof}

This completes the proof of the corollary.
\end{proof}

\section{Testing Support Size of Functions}
\label{section:functions}

For any $n \in \bN$, we write $\cH_n$ for the set of functions $f \colon \bN \to \zo$ which
satisfy $|f^{-1}(1)| \leq n$.
We will refer to pairs $(f, p)$, consisting of a function $f \colon \bN \to \zo$ and a probability
distribution $p$ over $\bN$, as \emph{function-distribution} pairs. A function-distribution pair
$(f,p)$ is \emph{$\epsilon$-far} from $\cH_n$ if
\[
  \forall h \in \cH_n \;:\; \Pru{\bm x \sim p}{ f(x) \neq h(x) } \geq \epsilon \,.
\]
For any multiset $S \subset \bN$ and function $f \colon \bN \to \zo$, we will write $S_f$ for the
labeled multiset
\[
  S_f \define \{ (x, f(x)) \;\colon\; x \in S \} \,.
\]

\begin{definition}[Testing Support Size of Functions]
\label{def:testing-functions}
A \emph{support-size tester for functions}, with sample complexity
$m(n,\epsilon, \successprob)$, is an algorithm $B$
which receives as input the parameters $n \in \bN$, $\epsilon \in (0,1)$, and $\successprob \in
(0,1)$, and is required to satisfy the following. For any function-distribution pair $(f, p)$
consisting of a function $f \colon \bN \to \zo$ and a probability distribution $p$ over $\bN$,
if $\bm{S}_f$ is a labeled sample of size $m = m(n,\epsilon,\successprob)$ drawn from $p$, then
\begin{enumerate}
\item If $f \in \cH_n$ then $\Pr{ B(\bm{S}_f) \text{ outputs } \ACCEPT } \geq \successprob$; and
\item If $(f,p)$ is $\epsilon$-far from $\cH_n$ then $\Pr{ B(\bm{S}_f) \text{ outputs } \REJECT }
\geq \successprob$.
\end{enumerate}
We write $m^\FUN(n,\epsilon,\successprob)$ for the optimal sample complexity of a support-size
tester for functions.
\end{definition}

We will also write $m^\DIST(n,\epsilon,\sigma)$ for the optimal sample complexity of support-size
testing of distributions, with parameters $n$ and $\epsilon$, and success probability $\successprob$
(replacing $3/4$ in \cref{def:testing-support-size}). The next two propositions establish
\cref{res:intro-functions} from the introduction.

\begin{proposition}
\label{res:dist-to-fun}
For any $n \in \bN$ and $\epsilon \in (0,1)$, and any $\successprob \in (0,1)$,
\[
  m^\DIST(n,\epsilon, \successprob) \leq m^\FUN(n,\epsilon, \successprob) \,.
\]
\end{proposition}
\begin{proof}
We reduce testing support size of distributions, to testing support size of functions. Let $B$ be a
distribution-free sample-based support-size tester for functions on domain $\bN$, with sample
complexity $m = m^\FUN(n,\epsilon)$. Then, given $n$ and $\epsilon$, the support-size tester for input
distribution $p$ is as follows.
\begin{enumerate}
\item Let $\bm S$ be $m$ independent random samples from $p$.
\item Output $B({\bm S}_f)$, where the sample $\bm S$ is labeled by the constant function $f(x)=1$.
\end{enumerate}
Suppose that $p$ is $\epsilon$-far from having $|\supp(p)| \leq n$. Then for any function $g :
\bN \to \zo$ with $|g^{-1}(1)| \leq n$, it must be the case that $\sum_{i \notin g^{-1}(1)} p_i \geq
\epsilon$.  So
\[
  \Pru{\bm S}{B(\bm{S}_f) \text{ outputs } \REJECT} \geq 3/4 \,.
\]
Now suppose $p$ has $|\supp(p)| \leq n$. Let $g \colon \bN \to \zo$ be the function $g(i) = \ind{i
\in \supp(p)}$, so $|g^{-1}(1)| \leq n$. Then
\[
  \Pru{\bm S}{B(\bm{S}_g) \text{ outputs } \ACCEPT} \geq 3/4 \,.
\]
Since $\sum_{i \notin g^{-1}(1)} p_i = 0$, the labeled sample $\bm{S}_f$ has the same distribution
as the labeled sample $\bm{S}_g$. So
\[
  \Pru{\bm S}{B(\bm{S}_f) \text{ outputs } \ACCEPT} \geq 3/4 \,. \qedhere
\]
\end{proof}

\begin{proposition}
\label{res:fun-to-dist}
For any $n \in \bN$ and $\epsilon \in (0,1)$, and any $\successprob, \xi \in (0,1)$ which satisfy
$\successprob + \xi < 1$,
\[
  m^\FUN(n,\epsilon, \successprob) \leq m^\DIST(n,\epsilon, \successprob + \xi) +
  \tfrac{\ln(2/\xi)}{\epsilon} \,.
\]
\end{proposition}
\begin{proof}
Let $A$ be the distribution support-size tester with sample complexity $m^\DIST$.
On input parameters $n$ and $\epsilon$, given sample access to the function-distribuiton pair
$(f,p)$, we define the tester $B$ which performs the following.
\begin{enumerate}
\item Take $m_1 = \tfrac{\ln(2/\xi)}{\epsilon}$ random labeled samples
$\bm{S}^{(1)}_f$. If $\bm{S}^{(1)}_f \cap f^{-1}(1) = \emptyset$, output $\ACCEPT$.
\item Let $\bm{z} \sim \bm{S}^{(1)} \cap f^{-1}(1)$ be chosen uniformly at random from the elements of the multiset
$\bm{S}^{(1)} \cap f^{-1}(1)$.
\item Take $m_2 = m^\DIST(n,\epsilon,\successprob + \xi)$ random labeled samples $\bm{S}^{(2)}_f$
and let $\bm{T}$ be the multiset obtained by replacing each element of $\bm{S}^{(2)} \cap f^{-1}(0)$
with $z$; \ie
\[
\bm{T} \define ( \bm{S}^{(2)} \cap f^{-1}(1) ) \cup \bm{Z}
\]
where  $\bm{Z}$ is the multiset containing $\bm{z}$ with multiplicity $|\bm{S}^{(2)} \cap
f^{-1}(0)|$.
\item Output $A(\bm{T})$.
\end{enumerate}
To establish correctness of this tester, note that, conditional on the event $E \define (\bm{S}^{(1)} \cap
f^{-1}(1) \neq \emptyset)$, $\bm{T}$ is distributed as $m_2$ independent samples from the
distribution $p^{(\bm z)}$ defined as
\[
  p^{(\bm z)}_i \define \begin{cases}
    p_i &\text{ if } i \in f^{-1}(1) \setminus \{\bm z\} \\
    p_i + \sum_{j \in f^{-1}(0)} p_j &\text{ if } i = \bm z \\
    0 &\text{ if } i \in f^{-1}(0) \,.
  \end{cases}
\]
If $|f^{-1}(1)| \leq n$ then it is clear that $|\supp(p^{(\bm z)})| \leq n$. If $(f,p)$ is
$\epsilon$-far from satisfying $|f^{-1}(1)| \leq n$, then the sum of the largest $n$ values $p_i$ on
domain $i \in f^{-1}(1)$ is at most $\sum_{j \in f^{-1}(1)} p_j -\epsilon$.  Therefore the sum of
the largest $n$ values of $p^{(\bm z)}_i$ is at most
\[
  \sum_{j \in f^{-1}(0)} p_j + \sum_{j \in f^{-1}(1)} p_j - \epsilon \leq 1 -\epsilon \,,
\]
so $p^{(\bm z)}$ is $\epsilon$-far from having $|\supp(p^{(\bm z)})| \leq n$.

Suppose that $|f^{-1}(1)| \leq n$. Write $\rho_E \define \Pru{\bm{S}^{(1)}}{E}$. Then
\[
  \Pru{\bm{S}^{(1)}, \bm{S}^{(2)}}{ B \text{ outputs } \ACCEPT }
  = (1-\rho_E) + \rho_E \cdot \Pru{\bm{S}^{(2)}}{ A(\bm{S}^{(2)}) \text{ outputs } \ACCEPT }
  \geq \successprob + \xi \,.
\]
Now suppose $(f,p)$ is $\epsilon$-far from satisfying $|f^{-1}(1)| \leq n$. Then
\[
  \Pru{\bm{S}^{(1)}, \bm{S}^{(2)}}{ B \text{ outputs } \REJECT }
  = \rho_E \cdot \Pru{\bm{S}^{(2)}}{ A(\bm{S}^{(2)}) \text{ outputs } \REJECT }
  \geq \rho_E \cdot (\successprob + \xi) \,.
\]
Since $(f, p)$ is $\epsilon$-far from satisfying $|f^{-1}(1)| \leq n$, it must be the case that
$\sum_{j \in f^{-1}(1)} p_j \geq \epsilon$, so
\[
  \rho_E \geq 1 - (1-\epsilon)^{m_1} \geq 1 - e^{-\epsilon m_1} \geq 1 - \xi/2 \,.
\]
Then
\[
  \Pru{\bm{S}^{(1)}, \bm{S}^{(2)}}{ B \text{ outputs } \REJECT }
  \geq (1-\xi/2) (\successprob + \xi)  \geq \successprob \,.\qedhere
\]
\end{proof}

%


 \newpage
   \printbibliography

@preamble{{
\providecommand{\FOCS}{Proceedings of the IEEE Symposium on Foundations of Computer Science (FOCS)}
\providecommand{\SODA}{Proceedings of the ACM-SIAM Symposium on Discrete Algorithms (SODA)}
\providecommand{\STOC}{Proceedings of the ACM SIGACT Symposium on Theory of Computing (STOC)}
\providecommand{\ITCS}{Proceedings of the Innovations in Theoretical Computer Science
  Conference (ITCS)}
\providecommand{\ICALP}{Proceedings of the International Colloquium on Automata, Languages, and Programming (ICALP)}
\providecommand{\ICML}{Proceedings of the International Conference on Machine Learning (ICML)}
\providecommand{\COLT}{Proceedings of the Conference on Learning Theory (COLT)}
\providecommand{\AISTATS}{Proceedings of the International Conference on Artificial Intelligence and Statistics (AISTATS)}
\providecommand{\TOCT}{ACM Transactions on Computation Theory (TOCT)}
\providecommand{\RANDOM}{Approximation, Randomization, and Combinatorial Optimization. Algorithms and Techniques (APPROX/RANDOM)}
\providecommand{\JACM}{Journal of the ACM (JACM)}
\providecommand{\SIAMJOC}{SIAM Journal on Computing}
\providecommand{\TOC}{Theory of Computing}
\providecommand{\TOIT}{IEEE Transactions on Information Theory}
\providecommand{\COLT}{Conference on Learning Theory (COLT)}
\providecommand{\NEURIPS}{Advances in Neural Information Processing Systems (NeurIPS)}
\providecommand{\JMLR}{Journal of Machine Learning Research}
\providecommand{\SOSA}{Symposium on Simplicity in Algorithms (SOSA)}
}}

@article{BF93,
  title={Estimating the number of species: a review},
  author={Bunge, John and Fitzpatrick, Michael},
  journal={Journal of the American statistical Association},
  volume={88},
  number={421},
  pages={364--373},
  year={1993},
  publisher={Taylor \& Francis},
  doi={10.1080/01621459.1993.10594330},
}

@article{FCW43,
  title={The relation between the number of species and the number of individuals in a random sample of an animal population},
  author={Fisher, Ronald A and Corbet, A Steven and Williams, Carrington B},
  journal={The Journal of Animal Ecology},
  pages={42--58},
  year={1943},
  publisher={JSTOR},
  doi={10.2307/1411},
}

@article{Goodman49,
  title={On the estimation of the number of classes in a population},
  author={Goodman, Leo A},
  journal={The Annals of Mathematical Statistics},
  volume={20},
  number={4},
  pages={572--579},
  year={1949},
  publisher={Institute of Mathematical Statistics},
  doi={10.1214/aoms/1177729949},
}

@article{Good53,
  title={The Population Frequencies of Species and the Estimation of Population Parameters},
  author={Good, {I.J.}},
  journal={Biometrika},
  pages={237--264},
  year={1953},
  publisher={JSTOR},
  doi={10.2307/2333344},
}

@inproceedings{ADOS17,
  title = 	 {A Unified Maximum Likelihood Approach for Estimating Symmetric Properties of Discrete Distributions},
  author =       {Jayadev Acharya and Hirakendu Das and Alon Orlitsky and Ananda Theertha Suresh},
  booktitle = 	 {\ICML},
  pages = 	 {11--21},
  year = 	 {2017},
  volume = 	 {70},
   url          = {http://proceedings.mlr.press/v70/acharya17a.html},
}

@article{FH23,
  title={Distribution testing under the parity trace},
  author={{Ferreira Pinto Jr.}, Renato and Harms, Nathaniel},
  year={2023},
  journal = {arXiv:2304.01374},
  arxivId = {2304.01374},
  eprint = {2304.01374},
  url = {https://arxiv.org/abs/2304.01374}
}

@misc{Han19,
  author={Hanneke, Steve},
  title={{CSTheory} StackExchange answer},
  year={2019},
  url = {https://cstheory.stackexchange.com/questions/40161/proper-pac-learning-vc-dimension-bounds}}

@article{BEHW89,
  title={Learnability and the {Vapnik-Chervonenkis} dimension},
  author={Blumer, Anselm and Ehrenfeucht, Andrzej and Haussler, David and Warmuth, Manfred K},
  journal={Journal of the ACM (JACM)},
  volume={36},
  number={4},
  pages={929--965},
  year={1989},
  publisher={ACM New York, NY, USA},
  doi={10.1145/76359.76371},
}

@article{Han16,
  title={The optimal sample complexity of {PAC} learning},
  author={Hanneke, Steve},
  journal={\JMLR},
  volume={17},
  number={38},
  pages={1--15},
  year={2016},
  url = {https://jmlr.org/papers/v17/15-389.html},
}

@inproceedings{HO19,
  title={Unified sample-optimal property estimation in near-linear time},
  author={Hao, Yi and Orlitsky, Alon},
  booktitle={\NEURIPS},
  year={2019},
  url = {https://proceedings.neurips.cc/paper/2019/hash/800b03685c22049f049801f6841861a2-Abstract.html},
}

@article{HO19pml,
  title={The broad optimality of profile maximum likelihood},
  author={Hao, Yi and Orlitsky, Alon},
  journal={\NEURIPS},
  year={2019},
  url = {https://proceedings.neurips.cc/paper/2019/hash/f9fd5ec4c141a95257aa99ef1b590672-Abstract.html},
}

@article{GGR98,
  title={Property testing and its connection to learning and approximation},
  author={Goldreich, Oded and Goldwasser, Shafi and Ron, Dana},
  journal={\JACM},
  volume={45},
  number={4},
  pages={653--750},
  year={1998},
  publisher={ACM},
  doi={10.1145/285055.285060},
}

@article{GR16,
  title={On sample-based testers},
  author={Goldreich, Oded and Ron, Dana},
  journal={\TOCT},
  volume={8},
  number={2},
  pages={1--54},
  year={2016},
  publisher={ACM},
  doi={10.1145/2898355},
}

@article{GR23,
  title={Testing distributions of huge objects},
  author={Goldreich, Oded and Ron, Dana},
  journal={TheoretiCS},
  volume={2},
  year={2023},
  publisher={Episciences.org},
  doi={10.46298/theoretics.23.12},
}

@inproceedings{HJW18,
  title={Local moment matching: A unified methodology for symmetric functional estimation and
distribution estimation under {Wasserstein} distance},
  author={Han, Yanjun and Jiao, Jiantao and Weissman, Tsachy},
  booktitle={\COLT},
  pages={3189--3221},
  year={2018},
  organization={PMLR},
  url = {http://proceedings.mlr.press/v75/han18b.html},
}

@article{WY19,
  title={Chebyshev polynomials, moment matching, and optimal estimation of the unseen},
  author={Wu, Yihong and Yang, Pengkun},
  journal={The Annals of Statistics},
  volume={47},
  number={2},
  pages={857--883},
  year={2019},
  doi={10.1214/17-AOS1665},
}

@inproceedings{AFL24,
  title={Support testing in the huge object model},
  author={Adar, Tomer and Fischer, Eldar and Levi, Amit},
  year={2024},
  booktitle={\RANDOM},
  doi={10.4230/LIPIcs.APPROX/RANDOM.2024.46},
}

@inproceedings{AF24,
  title={Refining the adaptivity notion in the huge object model},
  author={Adar, Tomer and Fischer, Eldar},
  year={2024},
  booktitle={\RANDOM},
  doi={10.4230/LIPIcs.APPROX/RANDOM.2024.45},
}

@inproceedings{VV11focs,
  title={The power of linear estimators},
  author={Valiant, Gregory and Valiant, Paul},
  booktitle={\FOCS},
  year={2011},
  doi={10.1109/FOCS.2011.81},
}

@inproceedings{VV11stoc,
  title={Estimating the unseen: an $n/\log (n)$-sample estimator for entropy and support size, shown
optimal via new {CLTs}},
  author={Valiant, Gregory and Valiant, Paul},
  booktitle={\STOC},
  year={2011},
  doi={10.1145/1993636.1993727},
}

@book{WY20,
  title={Polynomial methods in statistical inference: theory and practice},
  author={Wu, Yihong and Yang, Pengkun},
  year={2020},
  series={Foundations and Trends in Communications and Information Theory},
  publisher={now Publishers},
  doi={10.1561/0100000095},
}

@article{RRSS09,
  title={Strong lower bounds for approximating distribution support size and the distinct elements problem},
  author={Raskhodnikova, Sofya and Ron, Dana and Shpilka, Amir and Smith, Adam},
  journal={\SIAMJOC},
  volume={39},
  number={3},
  pages={813--842},
  year={2009},
  publisher={SIAM},
  doi={10.1137/070701649},
}

@inproceedings{RR20,
  title={Almost optimal distribution-free sample-based testing of k-modality},
  author={Ron, Dana and Rosin, Asaf},
  booktitle={\RANDOM}, 
  year={2020},
  doi={10.4230/LIPIcs.APPROX/RANDOM.2020.27},
}

@article{RR22,
  title={Optimal distribution-free sample-based testing of subsequence-freeness with one-sided error},
  author={Ron, Dana and Rosin, Asaf},
  journal={\TOCT},
  volume={14},
  number={1},
  pages={1--31},
  year={2022},
  publisher={ACM},
  doi={10.1145/3512750},
}

@article{VV17,
  title={Estimating the unseen: improved estimators for entropy and other properties},
  author={Valiant, Gregory and Valiant, Paul},
  year={2017},
  journal={\JACM},
  doi={10.1145/3125643},
}

@book{Gol17,
  title={Introduction to property testing},
  author={Goldreich, Oded},
  year={2017},
  publisher={Cambridge University Press},
  doi={10.1017/9781108135252},
}

@article{Can20,
  title={A survey on distribution testing: Your data is big. But is it blue?},
  author={Canonne, Cl{\'e}ment},
  journal={\TOC},
  year={2020},
  publisher={Theory of Computing Exchange},
  doi={10.4086/toc.gs.2020.009},
}

@misc{oeis_cheb_coeffs,
  title = {{OEIS} sequence {A008310}},
  author={{OEIS}},
 url= {https://oeis.org/A008310} 
}

@inproceedings{BFH21,
  title={{VC} dimension and distribution-free sample-based testing},
  author={Blais, Eric and {Ferreira Pinto Jr}, Renato and Harms, Nathaniel},
  booktitle={\STOC},
  year={2021},
  doi={10.1145/3406325.3451104},
}

@article{BCG19,
  title={Distribution testing lower bounds via reductions from communication complexity},
  author={Blais, Eric and Canonne, Clément L and Gur, Tom},
  journal={\TOCT},
  volume={11},
  number={2},
  pages={1--37},
  year={2019},
  publisher={ACM New York, NY, USA},
  doi={10.1145/3305270},
}

@inproceedings{CDS18,
  title={Testing for families of distributions via the Fourier transform},
  author={Canonne, Cl{\'e}ment L and Diakonikolas, Ilias and Stewart, Alistair},
  booktitle={\NEURIPS},
  year={2018},
  url={https://proceedings.neurips.cc/paper/2018/hash/aa8fdbb7d8159b3048daca36fe5c06d2-Abstract.html},
}

@inproceedings{VV16,
  title={Instance optimal learning of discrete distributions},
  author={Valiant, Gregory and Valiant, Paul},
  booktitle={\STOC},
  year={2016},
  doi={10.1145/2897518.2897641},
}

@incollection{GR20,
  author = {Goldreich, Oded and Ron, Dana},
  title = {On the Relation Between the Relative Earth Mover Distance and the Variation Distance (an
           Exposition)},
  booktitle = {Computational Complexity and Property Testing: On the Interplay Between Randomness
               and Computation},
  publisher = {Springer Cham},
  year      = {2020},
  editor    = {Goldreich, Oded},
  pages     = {141--151},
  doi={10.1007/978-3-030-43662-9_9},
}

@article{Gol19vdf,
  title={Testing Bipartitness in an Augmented {VDF} Bounded-Degree Graph Model},
  author={Goldreich, Oded},
  year={2019},
  journal={arXiv:1905.03070},
url = {https://arxiv.org/abs/1905.03070}
}

@incollection{Gol25,
    author={Goldreich, Oded},
    editor={Goldreich, Oded},
    title={On the Complexity of Estimating the Effective Support Size},
    booktitle={Computational Complexity and Local Algorithms : On the Interplay Between Randomness and Computation},
    year={2025},
    publisher={Springer Nature Switzerland},
    address={Cham},
    pages={204--236},
    isbn={978-3-031-88946-2},
    doi={10.1007/978-3-031-88946-2_13},
    url={https://doi.org/10.1007/978-3-031-88946-2_13}
}

@inproceedings{NT23,
  title={Estimating the Effective Support Size in Constant Query Complexity},
  author={Narayanan, Shyam and T{\v{e}}tek, Jakub},
  booktitle={\SOSA},
  pages={242--252},
  year={2023},
  doi={10.1137/1.9781611977585.ch22},
}

@inproceedings{KLR25,
  author =	{Kelman, Esty and Linder, Ephraim and Raskhodnikova, Sofya},
  title =	{Online Versus Offline Adversaries in Property Testing},
  booktitle =	{\ITCS},
  pages =	{65:1--65:18},
  year =	{2025},
  doi={10.4230/LIPIcs.ITCS.2025.65},
}

@article{PRR06,
  title={Tolerant property testing and distance approximation},
  author={Parnas, Michal and Ron, Dana and Rubinfeld, Ronitt},
  journal={Journal of Computer and System Sciences},
  volume={72},
  number={6},
  pages={1012--1042},
  year={2006},
  doi={10.1016/j.jcss.2006.03.002},
}

 \newpage
  
\appendix

\section{Coefficients of \texorpdfstring{$P_d$}{Pd} and Values of \texorpdfstring{$f$}{f}}
\label{section:coefficients}

\begin{proposition}[Chebyshev Polynomial Coefficients \cite{oeis_cheb_coeffs}]
\label{res:coefficients-formula}
The $j^{th}$ coefficient $b_j$ of $T_d(x)$ is 0 when $j$ has opposite parity as $d$.  Otherwise
\[
  b_j = 2^{j-1} \cdot d \cdot (-1)^{\frac{d-j}{2}}
    \cdot \frac{\left(\frac{d+j}{2}-1\right)!}{\left(\frac{d-j}{2}\right)! \cdot j!} \,.
\]
\end{proposition}

We put a bound on the largest coefficient of $T_d(x)$. The leading coefficient is $2^{d-1}$, so the
bound will necessarily be exponential in $d$.

\begin{proposition}
\label{res:cheb_coefficient_ub}
For every $d \geq 1$ and every $j \leq d$, the $j^\text{th}$ coefficient of $T_d(x)$
satisfies $|b_j| \leq d \cdot 3^d \leq 9^d$.
\end{proposition}
\begin{proof}
Let $b^{(d)}_j$ be the coefficient of $x^j$ in $T_d(x)$, with $b^{(d)}_{-1} \define 0$. We will
show by induction that for every $d \ge 0$, $|b^{(d)}_j| \leq \max\{1,d\} \cdot 3^d$. In the base
case we have $b^{(0)}_j, b^{(1)}_j \in \zo$. By the recursive definition of $T_d(x)$,
\begin{align*}
  b^{(d)}_j &= 2 \cdot b^{(d-1)}_{j-1} - b^{(d-2)}_j \\
            &\leq 2 \cdot |b^{(d-1)}_{j-1}| + |b^{(d-2)}_j| \\
            &\leq 2 \cdot \max\{1,d-1\} \cdot 3^{d-1} + \max\{1,d-2\} \cdot 3^{d-2} \\
            &\leq 2 \cdot d 3^{d-1} + d 3^{d-2}
            = \left(2 + \frac{1}{3}\right) d 3^{d-1} \leq d 3^d \,.
\end{align*}
The same argument shows $b^{(d)}_j \geq - d 3^d$. Now $d \le 3^d$, so $|b^{(d)}_j| \leq 3^{2d} =
9^d$.
\end{proof}

We can now calculate the coefficients of $P_d$ in terms of $\delta$ and the coefficients $b_j$ of
$T_d$.

\begin{proposition}[Coefficients of $P_d$]
\label{res:coefficients_P}
For given $\ell, r$, the polynomial $P_d(x) = \left(\sum_{k=1}^d a_k x^k\right) - 1$ defined by
$P_d(x) \define -\delta \cdot T_d(\psi(x))$ has coefficients
\[
    a_k = (-1)^{k+1} \cdot 
    \delta \cdot 2^k \sum_{j=k}^d b_j \cdot \frac{1}{(r-\ell)^j} \cdot \binom{j}{k}
\cdot (r+\ell)^{j-k} \,,
\]
where $b_j$ is the $j^\text{th}$ coefficient of $T_d(x)$.
\end{proposition}
\begin{proof}
For convenience, write $\alpha \define r+\ell$ and $\beta \define r - \ell$.
Observe that
\begin{align*}
\frac{1}{\delta} \cdot P_d(x)
  = -T_d(\psi(x)) 
  &= -\sum_{j=0}^d b_j \cdot (\psi(x))^j \\
  &= -\sum_{j=0}^d b_j \left(-\frac{2x - r - \ell}{r - \ell}\right)^j 
   = -\sum_{j=0}^d b_j \left(\frac{-2x + \alpha}{\beta}\right)^j  \\
  &= -\sum_{j=0}^d b_j \sum_{k=0}^j \binom{j}{k}
        \left(\frac{-2x}{\beta}\right)^k \left(\frac{\alpha}{\beta}\right)^{j-k}.
\end{align*}
For each fixed $k$, the coefficient of $x^k$ in this polynomial is
\[
  -(-2)^k \sum_{j=k}^d b_j \cdot \frac{1}{\beta^j} \cdot \binom{j}{k} \alpha^{j-k} 
  = (-1)^{k+1} 2^k \sum_{j=k}^d b_j \cdot \frac{1}{(r-\ell)^j} \cdot \binom{j}{k} (r+\ell)^{j-k} \,.
  \qedhere
\]
\end{proof}

We combine these calculations into a formula for the values $f(k)$ in the test statistic
(\cref{def:test-statistic}), for the sake of completeness and to generate \cref{fig:f-values}.  A
similar figure appears in \cite{WY19} for their support size estimator.
\begin{figure}[h!]
\centering
\includegraphics[width=0.5\textwidth]{"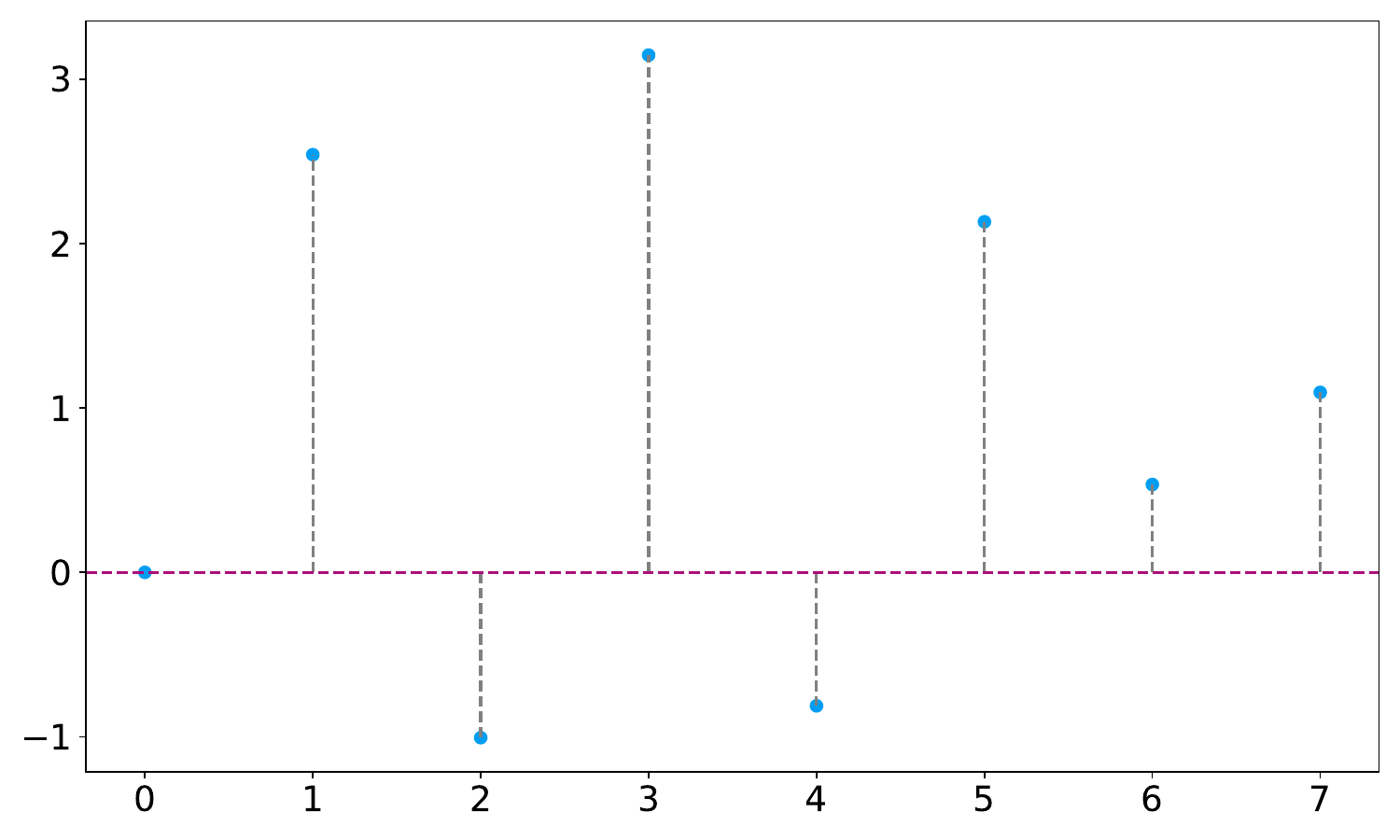"}
\caption{Example values of $1+f(\bm{N_i})$ in the estimator.} 
\label{fig:f-values}
\end{figure}

\begin{proposition}
\label{res:f-values}
For given $\ell, r, m$ and $d$, the function $f$ in \cref{def:test-statistic} obtained from the
polynomial $P_d$ satisfies $f(0) = -1$, $f(k) = 0$ for $k > d$, and for all $k \in [d]$,
\[
    f(k) = (-1)^{k+1} \cdot \delta \cdot d \cdot \frac{1}{m^k}
          \sum_{\substack{j=k \\ j \equiv d \mod 2}}^d (-1)^{\tfrac{d-j}{2}} \left(2^{k+j-1} \cdot 
\frac{\left(\tfrac{d+j}{2}-1\right)!}{\left(\tfrac{d-j}{2}\right)! (j-k)!}\right) \cdot
\frac{(r+\ell)^{j-k}}{(r-\ell)^j} \,.
\]
\end{proposition}
\begin{proof}
\begin{align*}
  f(k) &= a_k \frac{k!}{m^k}
        = (-1)^{k+1} \cdot \delta \cdot 2^k \frac{k!}{m^k}
          \sum_{j=k}^d b_j \cdot \frac{1}{(r-\ell)^j} \cdot \binom{j}{k} \cdot (r + \ell)^{j-k} \\
       &= (-1)^{k+1} \cdot \delta \cdot 2^k \frac{k!}{m^k}
          \sum_{\substack{j=k \\ j \equiv d \mod 2}}^d  (-1)^{\tfrac{d-j}{2}} \left(2^{j-1} \cdot d \cdot
\frac{\left(\tfrac{d+j}{2}-1\right)!}{\left(\tfrac{d-j}{2}\right)! \cdot j!}\right) \cdot
\frac{(r+\ell)^{j-k}}{(r-\ell)^j} \cdot \binom{j}{k}  \\
       &= (-1)^{k+1} \cdot \delta \cdot d \cdot \frac{1}{m^k}
          \sum_{\substack{j=k \\ j \equiv d \mod 2}}^d (-1)^{\tfrac{d-j}{2}} \left(2^{k+j-1} \cdot 
\frac{\left(\tfrac{d+j}{2}-1\right)!}{\left(\tfrac{d-j}{2}\right)! (j-k)!}\right) \cdot
\frac{(r+\ell)^{j-k}}{(r-\ell)^j} \,. \qedhere
\end{align*}
\end{proof}

\resmaxu*
\begin{proof}
Using \cref{res:coefficients_P} with $b_j$
being the $j^\text{th}$ coefficient of $T_d(x)$, and using $|b_j| \leq d 3^d$ from
\cref{res:cheb_coefficient_ub}, we have 
\begin{align*}
  |f(k)| \define |a_k| \frac{k!}{m^k}
      &= \delta \cdot 2^k \frac{k!}{m^k} \cdot \abs*{ \sum_{j=k}^d b_j \cdot \binom{j}{k}
          \cdot \frac{(r+\ell)^{j-k}}{(r-\ell)^j}} 
      \leq \delta \cdot 2^k \frac{k!}{m^k} \cdot \sum_{j=k}^d |b_j| \cdot \binom{j}{k}
          \cdot \frac{(r+\ell)^{j-k}}{(r-\ell)^j} \\
      &\leq \delta \cdot d \cdot 3^d \cdot 2^k \frac{k!}{m^k}
          \cdot \sum_{j=k}^d  \binom{j}{k} \cdot \frac{(r+\ell)^{j-k}}{(r-\ell)^j} \\
      &=  \delta \cdot d \cdot 3^d \cdot 2^k \frac{k!}{m^k (r-\ell)^k}
         \sum_{j=k}^d \binom{j}{k} \left(\frac{r+\ell}{r-\ell}\right)^{j-k} \\
      &=  \delta \cdot d \cdot 3^d \cdot 2^k \frac{1}{m^k (r-\ell)^k}
         \sum_{j=k}^d \frac{j!}{(j-k)!} \left(\frac{r+\ell}{r-\ell}\right)^{j-k} \\
      &\leq  \delta \cdot d \cdot 3^d \cdot 2^k \frac{1}{m^k (r-\ell)^k} (d-k+1) \frac{d!}{(d-k)!}
         \left(\frac{r+\ell}{r-\ell}\right)^{d-k} \\
      &\leq  \delta \cdot d^2 \cdot 3^d \cdot \left( \frac{2d}{m (r-\ell)} \right)^k
         \left(\frac{r+\ell}{r-\ell}\right)^{d-k} \,. \qedhere
\end{align*}
\end{proof}

\section{Derivative of \texorpdfstring{$T_d$}{Td}}
\label{section:td-prime}

The derivative of $T_d$ is often expressed in terms of the so-called \emph{Chebyshev polynomials of
the second kind}. However, since we only require a simple lower bound, we opt for a direct
calculation instead.

\restdprimelb*
\begin{proof}
    Note that $T_d'(1+\gamma) = \frac{\odif}{\odif \gamma} T_d(1+\gamma)$. Using the closed-form
    formula for $T_d$, we have
    \begin{align*}
        \frac{\odif}{\odif \gamma} T_d(1+\gamma)
        &= \frac{1}{2} \left[
            \frac{\odif}{\odif \gamma} \left(1+\gamma + \sqrt{2\gamma+\gamma^2}\right)^d
            + \frac{\odif}{\odif \gamma} \left(1+\gamma - \sqrt{2\gamma+\gamma^2}\right)^d
        \right] \\
        &= \frac{1}{2} \Bigg[
            d \left(1+\gamma + \sqrt{2\gamma+\gamma^2}\right)^{d-1} \left(
                1 + \frac{2+2\gamma}{2\sqrt{2\gamma+\gamma^2}}
            \right) \\
        &\qquad \quad
            + d \left(1+\gamma - \sqrt{2\gamma+\gamma^2}\right)^{d-1} \left(
                1 - \frac{2+2\gamma}{2\sqrt{2\gamma+\gamma^2}}
            \right)
        \Bigg] \\
        &= \frac{d}{2\sqrt{2\gamma+\gamma^2}} \left[
            \left(1+\gamma + \sqrt{2\gamma+\gamma^2}\right)^d
            - \left(1+\gamma - \sqrt{2\gamma+\gamma^2}\right)^d
        \right] \\
        &\ge \frac{d}{\sqrt{3\gamma}} \left[
            \frac{\left(1+\gamma + \sqrt{2\gamma+\gamma^2}\right)^d
                + \left(1+\gamma - \sqrt{2\gamma+\gamma^2}\right)^d}{2}
            - \left(1+\gamma - \sqrt{2\gamma+\gamma^2}\right)^d
        \right] \\
        &\ge \frac{d}{\sqrt{3\gamma}} \left( T_d(1+\gamma) - 1 \right) \,,
    \end{align*}
    where the inequalities used the fact that $\gamma \in (0,1)$ and hence $1 + \gamma -
    \sqrt{2\gamma+\gamma^2} \in (0, 1)$.
\end{proof}

\section{Remark on the Bounded Support Size Assumption}
\label{section:support-assumption}
\label{section:instance-optimal}

As noted in \cite[pp.~21]{GR23}, there seems to be no simple way to obtain a support-size tester
with sample complexity $O\left(\frac{n}{\epsilon^2 \log n}\right)$ from the prior literature,
without a promise that the true support size is at most $O(n)$. By ``simple''\!, we mean using the
known results as a black box without carefully analyzing how the algorithms work. Let us elaborate
on this.

\subsection{Using histogram learner with TV distance guarantee.}
Prior works \cite{VV11stoc,VV17,HJW18} prove the following (informal) statement: There is an algorithm $A$
which draws $m$ samples from an arbitrary probability distribution $p$
and outputs a sorted distribution $q$ with $q_1 \geq q_2 \geq \dotsm$, with the following guarantee.  If
the support size of $p$ satisfies
\[
  |\supp(p)| \leq O\left( \epsilon^2 m \log m \right),
\]
then the output $q$ satisfies $\|q - p^* \|_1 < \epsilon$ with probability at least $3/4$, where
$p^*$ is the sorted copy of $p$, \ie it is a permutation of $p$ such that $p^*_1 \geq p^*_2 \geq
\dotsm$. In particular, this will hold if $|\supp(p)| \leq n$ and we draw $m = \Theta(
\tfrac{n}{\epsilon^2 \log n})$ samples. Using the typical testing-by-learning technique, this
suggests the following tester:
\begin{enumerate}
\item Draw $m = O( \tfrac{n}{\epsilon^2 \log n} )$ samples and produce a sorted distribution
$q$. If input $p$ satisfies $|\supp(p)| \leq n$ then (with high probability) we have
$\|q - p^*\|_1 < \epsilon/2$. Output $\REJECT$ if $q$ is not $\epsilon/2$-close to
having support size at most $n$.
\item Verify that $q$ is indeed close to $p^*$ and output $\ACCEPT$ if this is the case. In
particular, using a \emph{tolerant}\footnote{A tolerant tester should $\ACCEPT$ if the distribution
$p$ is $\epsilon_1$-close to the class, and $\REJECT$ if $p$ is $\epsilon_2$-far from the set, for
$\epsilon_1 < \epsilon_2$.} tester to test whether $p$ is $\epsilon$-close to the set of
permutations of $q$. Unfortunately, the only tolerant tester for
this property that we are aware of is itself obtained by the same testing-by-learning approach of
approximating the histogram (see \eg \cite{Gol17}), which is only guaranteed to work when
$|\supp(p)| \leq n$.
\end{enumerate}
Algorithms for learning the histogram \emph{should} be able to test support size with no
promise on the true support size. One could hope for the following property of the algorithms: Given
$m = O\left(\tfrac{n}{\epsilon^2 \log n}\right)$ samples, if $p$ is $\epsilon$-far from having
support size at most $n$, then with high probability the output $q$ is $\Omega(\epsilon)$-far from
having support size at most $n$. This does not appear to follow from the proven properties of these
algorithms without further analysis. 

\subsection{Using the histogram learner with relative Earth-mover distance guarantee.}
In \cite{VV16}, the assumption of bounded support size is removed from the histogram learner.
However, by necessity, this new algorithm's guarantee is not in terms of the TV distance but rather
the \emph{truncated relative Earth-mover distance} $R_\tau(\cdot, \cdot)$. We will not formally
define this metric here, but we note two of its properties:
\begin{enumerate}
\item $R_\tau$ ignores densities less than $\tau$, \ie if $p, p'$ and  $q, q'$ are pairs of distributions
such that $p_i = p'_i$ whenever $p_i, p'_i > \tau$, and $q_i = q'_i$ whenever $q_i, q'_i > \tau$,
then $R_\tau(p,q) = R_\tau(p',q')$.
\item $\dist_\TV(p^*,q^*) \leq R_0(p,q)$ (see also the exposition \cite{GR20}).
\end{enumerate}
The main result in \cite{VV16} is

\begin{theorem}[Theorem~2 of {\cite{VV16}}]
    \label{res:vv16-main}
    There exists an algorithm satisfying the following for some absolute constant $c$, sufficiently
    large $m$, and any $w \in [1, \log m]$. Given $m$ independent draws from a distribution $p$ with
    histogram $p^*$, with high probability the algorithm outputs a generalized\footnote{A standard
        histogram maps each density value $\alpha$ to the number of elements $i$ satisfying $p_i =
    \alpha$. A generalized histogram may map $\alpha$ to non-integral values as well.} histogram
    $q^*$ satisfying
    \[
        R_{\frac{w}{m \log m}}(p^*, q^*) \le \frac{c}{\sqrt{w}} \,.
    \]
\end{theorem}

As an application, \cite{VV16} give a procedure to estimate the expected number of unique elements
that would be seen in a sample of size $m \log m$, given a sample of size $m$:

\begin{theorem}[Proposition~1 of {\cite{VV16}}]
    \label{res:vv16-unique}
    Given $m$ samples from an arbitrary distribution $p$, with high probability over the randomness
    of the samples, one can estimate the expected number of unique elements that would be seen in a
    set of $k$ samples drawn from $p$, to within error $k \cdot c \sqrt{\frac{k}{m \log m}}$ for
    some universal constant $c$.
\end{theorem}

We sketch two ways to obtain a support-size tester with sample complexity $\poly(1/\epsilon) \cdot
O(\tfrac{n}{\log n})$ using these theorems.

\paragraph*{1. Estimating Unique Elements.} Combining \cref{res:vv16-unique} with
the na\"ive tester from \cref{section:small-epsilon}, which makes its decision based on the number of
unique elements observed, gives a support size tester with sample complexity $O(\frac{n}{\epsilon^5
\log n})$. The idea is to set $m \define \Theta(\frac{n}{\epsilon^5 \log n})$, use
\cref{res:vv16-unique} to obtain an estimate of the expected number of unique elements that would be
seen in a sample of size $k \define \Theta(n/\epsilon)$ to within error $\epsilon n / 8$, and accept
if and only if this estimate is at most $(1 + \epsilon/8) n$.

When $|\supp(p)| \le n$, the expected number of unique elements is at most $n$, so the tester
accepts. Now, suppose $p$ is $\epsilon$-far from being supported on $n$ elements. We claim that any
set of size at most $(1 + \epsilon/2) n$ misses at least $\epsilon / 2$ mass from $p$. Note that the
$n^{\text{th}}$ largest probability mass in $p$ is at most $1/n$, and the remaining elements of
smaller probability mass (the \emph{light elements}) make up at least $\epsilon$ mass; thus there
are at least $\epsilon n$ light elements. Hence the most probability mass that a set of size
$(1+\epsilon/2) n$ can cover comes from picking the $n$ elements of largest mass plus $\epsilon n /
2$ light elements, for a total of at most $1 - \epsilon + \frac{\epsilon n}{2} \cdot \frac{1}{n} = 1
- \epsilon/2$ mass, \ie $\epsilon/2$ mass is missed.

Therefore, an argument similar to the proof of \cref{res:naive-estimator} shows that the expected
number of unique elements observed in a sample of size $\Theta(n/\epsilon)$ is at least $(1 +
\epsilon/4) n$, so the tester rejects.

\ignore{
\footnote{Let $\bm{X}_i$ indicate whether a new unique element was observed in
    the $i^{\text{th}}$ step, let $\bm{S}_i \define \sum_{j \le i} \bm{X}_i$ be the number of unique
    elements seen in the first $i$ steps, and let $\bm{Y}_i \define \ind{\bm{S}_{i-1} \ge
    (1+\epsilon/2)n \text{ or } \bm{X}_i = 1}$. Then for any $m \in \bN$, $\bm{S}_m \ge
    (1+\epsilon/2)n$ if and only if $\sum_{i \le m} \bm{Y}_i \ge (1+\epsilon/2)n$. Moreover, for any
    setting of $(\bm{Y}_j)_{j<i} = (y_j)_{j<i}$, the conditioned random variable $\bm{Y}_i \,|\,
    (\bm{Y}_j)_{j<i} = (y_j)_{j<i}$ stochastically dominates a Bernoulli random variable with
    parameter $\epsilon/2$ (because either $(1+\epsilon/2)n$ unique elements have been observed, or
    the unseen elements make up at least $\epsilon/2$ mass). Hence a Chernoff bound implies that,
    for $k = \Theta(n/\epsilon)$, $\sum_{i \le k} \bm{Y}_i \ge 2n$ with probability at least $1 -
e^{-\Omega(n)}$. Therefore $\Ex{\bm{S}_k} \ge (1 - e^{-\Omega(n)}) (1+\epsilon/2) n > (1+\epsilon/4)
n$.}
}

\paragraph*{2. Learning the Histogram.} Instead of applying the na\"ive tester on
top of the unique elements estimator, one may hope for a more efficient tester by directly analyzing
the learned histogram $q^*$ from \cref{res:vv16-main}. We expect that one could obtain a support
size tester with sample complexity $O(\frac{n}{\epsilon^3 \log n})$ as follows. By setting $w
\define \Theta(1/\epsilon^2)$ and $m \define \Theta(\frac{n}{\epsilon^3 \log n})$ (at least when
this satisfies the constraint $w \le \log m$ from \cref{res:vv16-main}), the algorithm from
\cref{res:vv16-main} yields $q^*$ satisfying
\begin{equation}
    \label{eq:remd-bound}
    R_{\Theta(\epsilon/n)}(p^*, q^*) \le \Theta(\epsilon) \,.
\end{equation}
It appears that this suffices for testing support size. However, we will not attempt to prove this,
since the bound is worse than our \cref{res:intro-main}, and it is not possible to improve this
bound while treating the algorithm as a black box. This is because, if $\tau \gg \epsilon/n$, we can
construct distributions $p,q$ whose histograms satisfy $R_\tau(p^*,q^*) < \epsilon$ and yet
$|\supp(p)| \leq n$ while $q$ is $\epsilon$-far from having $|\supp(p)| \leq n$.

\ignore{
We may construct probability distributions $p$ and $p'$ such that $p$ has support
size $n$, and in particular assigns probability $4\epsilon/n$ to each of $n/2$ elements (for a total
of $2\epsilon$ mass), while $p'$ is obtained from $p$ by splitting each of those $n/2$ elements into
two elements assigned equal probability $2\epsilon/n$, so that $p'$ is $\epsilon$-far from having
support size $n$. The (non-truncated) relative earth-mover distance between $p$ and $p'$ is
$\Theta(\epsilon)$ (and indeed so is the TV distance between them), suggesting that we cannot make
the right-hand side of \eqref{eq:remd-bound} any larger if we wish to distinguish between $p$ and
$p'$, \ie we should set $w = \Omega(1/\epsilon^2)$. Then, if the threshold $\tau$ for the truncated
relative earth-mover distance $R_\tau$ is $\tau \gg \epsilon/n$, the distributions $p$ and $p'$
would be indistinguishable under $R_\tau$, so we must have $\frac{w}{m \log m} = O(\epsilon/n)$.
Since $w = \Omega(1/\epsilon^2)$, we are forced to set $m = \Omega(\frac{n}{\epsilon^3 \log n})$.
}

\end{document}